\documentclass[aps,pra,twocolumn,
10pt,longbibliography,superscriptaddress,
floatfix]{revtex4-1}

\newif\ifarxiv
\newif\ifjournal
\journalfalse \arxivtrue  

\ifjournal
\usepackage{bibunits}
\defaultbibliography{signs}
\defaultbibliographystyle{./myapsrev4-1}
\fi

\usepackage[T1]{fontenc} 
\usepackage[english]{babel} 
\usepackage[utf8]{inputenc}

\usepackage{times}

\usepackage[table,usenames,dvipsnames]{xcolor}
\usepackage[
	colorlinks=true,
		hypertexnames=false
			]{hyperref}
\PassOptionsToPackage{linktocpage}{hyperref} 

\hypersetup{
  colorlinks   = true, 
  urlcolor = magenta,
  linkcolor    = blue!60!black, 
  citecolor=green!60!black 
}
\usepackage{bookmark} 
\usepackage{booktabs}

\usepackage{graphicx} 	

\usepackage{amsthm} 
\usepackage{amssymb}
\usepackage{amsmath,mathtools,bbm}
\usepackage{dsfont}
\usepackage{comment}
\usepackage{enumerate}
\usepackage{ragged2e}

\definecolor{applegreen}{rgb}{0.55, 0.71, 0.0}

\usepackage[most]{tcolorbox}
\newtcolorbox{mybox}[1][]{%
  boxsep = 0mm,
  colback      = applegreen!45!white,
  colframe     = applegreen,
  halign       = flush left,
  enhanced,
}

\usepackage{tikz}
\usepackage{complexity}
\newclass{\stoqma}{StoqMA}
\newclass{\classP}{P}
\newclass{\np}{NP}
\newclass{\sharpP}{\#P}
\newclass{\ma}{MA}
\newclass{\qma}{QMA}
\newclass{\am}{AM}
\newclass{\maxcut}{MAXCUT}
\newclass{\sat}{SAT}
\newclass{\maxtwosat}{MAX2SAT}
\newclass{\twosat}{2SAT}
\newclass{\threesat}{3SAT}
\newclass{\se}{Sign Easing}

\newtheorem{theorem}{Theorem}
\newtheorem{definition}[theorem]{Definition}
\newtheorem{lemma}[theorem]{Lemma}

\newtheorem{example}[theorem]{Example}

\DeclareMathOperator{\tr}{Tr}
\newcommand{\ket}[1]{\vert{#1}\rangle}
\newcommand{\bra}[1]{\langle{#1}\vert}
\newcommand{\braket}[2]{\langle{#1}\vert #2 \rangle}

\newcommand{\abs}[1]{\vert #1 \vert}

\newcommand{\norm}[1]{\Vert #1 \Vert}

\newcommand{\nonstoq}[1]{\ensuremath #1_{\neg}}
\newcommand{\stoqabs}[1]{( #1 - 2 \nonstoq{#1})}

\DeclareMathOperator{\sign}{sign}

\DeclareMathOperator{\var}{Var}

\DeclareMathOperator{\diag}{diag}

\newcommand{\mc}[1]{\mathcal #1}
\newcommand{\mb}[1]{\mathbb #1}

\renewcommand{\l}[0]{\langle}

\renewcommand{\r}[0]{\rangle}

\newcommand{\id}{\mathbbm{1}}
\newcommand{\ee}{\mathrm{e}}

\newcommand{\jmodel}{$J_0$-$J_1$-$J_2$-$J_3$}
\usepackage{xcolor}
\definecolor{jens}{rgb}{0,.8,.5}

\definecolor{daniel}{rgb}{.7,.1,0}

\definecolor{dominik}{rgb}{0.4,.0,0.6}

\definecolor{ingo}{rgb}{0.4,.5,0.6}

\makeatletter 
 \hypersetup{pdftitle = {Easing the Monte Carlo sign problem},
	     pdfauthor = {Dominik Hangleiter, 
	     Ingo Roth, Daniel Nagaj, Jens Eisert},
	     pdfsubject = {quantum monte carlo, 
	       	quantum physics,
	       	sign problem,
	       	complexity theory},
 	     pdfkeywords = {quantum,
 	      	monte carlo, 
 	     	sign problem,
 	     	maxcut,
 	     	xyz hamiltonian,
 	     	stoquastic,
 	     	path integral,
 	     	quantum physics,
 	     	complexity theory,
 	     	np complete
			    }
	    }
\makeatother

\newcommand{\fu}{Dahlem Center for Complex Quantum Systems, Freie Universit\"{a}t Berlin, Germany}
\usepackage[normalem]{ulem}

\newcommand{\detailsandproof}{the Supplementary Material~\cite{suppmaterial}}

\begin{document}
\ifjournal
\begin{bibunit}
\fi

\title{Easing the Monte Carlo sign problem}

\author{Dominik Hangleiter}
\email[Corresponding author: ]{dominik.hangleiter@fu-berlin.de}

\affiliation{\fu}

\author{Ingo Roth}
\affiliation{\fu}

\author{Daniel Nagaj}
\affiliation{RCQI, Institute of Physics, Slovak Academy of Sciences, Bratislava, Slovakia}

\author{Jens Eisert}
\affiliation{\fu}


\begin{abstract}
	Quantum Monte Carlo (QMC) methods are the gold standard for studying equilibrium properties of quantum many-body systems -- their phase transitions, ground and thermal state properties.
	However, in many interesting situations QMC methods are faced with a \emph{sign problem}, causing the severe limitation of an exponential increase in the sampling complexity and hence the run-time of the QMC algorithm. 
	In this work, we develop a systematic, generally applicable, and practically feasible methodology for \emph{easing the sign problem} by efficiently computable basis changes and use it to rigorously assess the sign problem.
	Our framework introduces measures of non-stoquasticity that -- as we demonstrate analytically and numerically -- at the same time provide a practically relevant and efficiently computable figure of merit for the severity of the sign problem. 
	We show that those measures can practically be brought to a good use to ease the sign problem. 
	To do so, we use geometric algorithms for optimization over the orthogonal group and ease the sign problem of frustrated Heisenberg ladders. 
	Complementing this pragmatic mindset, we prove that easing the sign problem in terms of those measures is in general an \np-complete task for nearest-neighbour Hamiltonians and simple basis choices by a polynomial reduction to the \maxcut-problem.
	Intriguingly, easing remains hard even in cases in which we can efficiently assert that no exact solution exists.
\end{abstract}


\maketitle

\ifjournal
\section*{Introduction}
\fi

\emph{Quantum Monte Carlo (QMC)} techniques are central to our understanding of the equilibrium physics of many-body quantum systems. They provide arguably one of the most powerful workhorses for efficiently calculating expectation values of observables in ground and thermal states of various classes of many-body Hamiltonians
\cite{MonteCarloOld,troyer_nonlocal_2003,ReviewMonteCarlo,MonteCarloValidator}.
For a Hamiltonian $H$ in dimension $D$, the idea at the heart of the most prominent variant of QMC is to sample out world lines in a corresponding $(D+1)$-dimensional system, where the additional dimension is the (Monte Carlo) time dimension.
These world lines correspond to paths through an $m$-fold expansion of $\ee^{-\beta H} = (\ee^{-\beta H/m})^m$ where an entry of $\ee^{-\beta H/m}$ in a local basis is selected in each step. 
Each such path is associated with a probability which is proportional to the product of the selected entries. 
To sample from the resulting distribution, one can construct a suitable Markov chain of paths satisfying detailed balance, which -- if gapped -- eventually converges to its equilibrium distribution representing the thermal state. 
Generally speaking, concentration-of-measure phenomena often make such a procedure efficient.

In the classical variant of Monte Carlo, the Hamiltonian is always diagonal, giving rise to positive weights. 
In QMC, in contrast, positive (in general even complex) off-diagonal matrix elements of $H$ potentially give rise to negative weights of the paths. 
This leads to what is famously known as the \emph{sign problem} of QMC, namely that now one is faced with the task of sampling a quasi-probability distribution (normalized but non-positive) as opposed to a non-negative probability distribution.
This task can be achieved by introducing a suitable probability distribution that reproduces the desired sampling averages but typically comes at the cost of an exponential increase in the sampling complexity and hence the runtime of the algorithm.
For example, in world-line Monte Carlo one takes the absolute value of the quasi-probability distribution and then computes the \emph{average sign} which is given by the expectation value of the signs of the quasi-probabilities with respect to the new distribution. 
The sign problem is particularly severe for fermionic Hamiltonians, as the particle-exchange anti-symmetry forces their matrix elements to have alternating signs in the standard basis. Naturally, though, it also appears for bosonic or spin Hamiltonians.
The sign problem therefore severely limits our understanding of quantum materials. 
One can go as far as seeing it to divide strongly correlated systems into easy and intractable cases. 

A basic but fundamental insight is that the QMC sign problem is a \emph{basis-dependent} property \cite{hatano_representation_1992,hastings_how_2015}. 
For this reason, saying that `a Hamiltonian does or does not exhibit a sign-problem' is meaningless without specifying a basis. 
Since physical quantities of interest are independent of the basis choice, the observation that the sign problem is basis-dependent gives immediate hope to actually mitigate the sign problem of QMC by expressing the Hamiltonian in a suitable basis. 
This is not guaranteed to improve the overall runtime of QMC as governed not only by the sampling complexity but also by the computational complexity of producing an individual sample. 
Nonetheless, mitigating the sign problem is widely expected to render QMC efficient in many situations. 

In this work, we establish a comprehensive novel framework for assessing, understanding, and optimizing the sign problem computationally, asking the questions: 
\emph{What is the optimal computationally meaningful local basis choice for a QMC simulation of a Hamiltonian problem, can we find it, and how hard is this task in general?}  

\ifarxiv 
\subsection*{Curing the sign problem}
\fi

In fact, it is known that one can completely \emph{cure} the sign problem using basis rotations in certain situations. 
For specific models, sign-problem free bases can be found analytically,  
involving non-local bases, for example by using so-called auxiliary-field \cite{wu_exact_2003}, Jordan-Wigner \cite{okunishi_symmetry-protected_2014} or Majorana \cite{li_solving_2015,li_majorana-time-reversal_2016} transformations. 
One can also exploit specific known properties of the system such as that the system dimerizes \cite{nakamura_vanishing_1998,alet_sign-problem-free_2016,honecker_thermodynamic_2016,wessel_efficient_2017}.
Such findings motivate the quest for a more broadly applicable systematic search for basis changes that avoid the sign problem, in a way that 
\emph{does not depend} on the specific physics of the problem at hand. 
After all, in a QMC simulation one wants to \emph{learn about the physics} of a system in the first place and, indeed, the optimal basis choice may very well be closely related to that physics. 

Clearly, a useful notion of curing has to restrict the set of allowed basis transformation such that expressing the Hamiltonian in the new basis is still computationally tractable. 
For example, in its eigenbasis every Hamiltonian is diagonal and thus sign-problem free, but even writing down this basis typically requires an exponential amount of resources. 
The \emph{intrinsic sign problem} of a Hamiltonian is thus a property of its \emph{equivalence classes} under conjugation with some suitable subgroup of the unitary group. 
The simplest examples of such choices include local Hadamard, Clifford or unitary transformations. 
Most generally, one can allow for quasi-local circuits which are efficiently computable~\cite{hastings_how_2015}, including short circuits and matrix product unitaries~\cite{cirac_matrix_2017,sahinoglu_matrix_2017}, but also invertible transformations~\cite{dobrautz_compact_2019}. 

A both useful and simple sufficient condition for the absence of a sign problem, independent of the specifics of a simulation, is that the Hamiltonian matrix is \emph{stoquastic}, i.e., has only non-positive off-diagonal entries. 
In fact, stoquasticity provides a useful framework to assess the computational complexity of a systematic approach to curing the sign problem~\cite{troyer_computational_2005}. 
Only recently has the curing problem, to decide whether a stoquastic local basis exists, been shown to be an \np-complete task under on-site unitary transformations for $2$-local Hamiltonians with additional local fields~\cite{marvian_computational_2018,klassen_hardness_2019}, while it remains efficiently solvable for strictly $2$-local Hamiltonians~\cite{klassen_two-local_2018,klassen_hardness_2019}. 
But any such approach is faced with the question: 
Is all hope lost for simulating a Hamiltonian problem via QMC more efficiently even when a stoquastic basis cannot be found in polynomial time?

\ifjournal
\section*{Results}
\fi
\subsection*{A pragmatic approach: Easing the sign problem}

This leads us to the first part of the initially posed question: what is the \emph{optimal computationally meaningful choice} of basis?
In any Monte Carlo algorithm, computational hardness due to a sign problem
is manifested in a super-polynomial increase in its 
sample complexity as the system size grows. 
Intuitively speaking, the sample complexity increases because the variance of the Monte Carlo estimator does. 
In this mindset, finding a QMC algorithm with feasible runtime for Hamiltonians with a sign problem does not require the much stronger task of finding a basis in which the Hamiltonian is fully stoquastic. 
Indeed, in many cases such a basis may not even exist within a given subgroup of the unitaries. 
Rather, often it is sufficient to merely find a basis in which the Hamiltonian is \emph{approximately stoquastic} so that the scaling of the variance of the corresponding estimator with the system size is more favourable -- ideally polynomial. 
More pragmatically still, practitioners in QMC are increasingly less worried about small sign problems for which simulations are still feasible for reasonable system sizes using state-of-the-art computing power. 
This remains true even if the sampling effort may strictly speaking diverge exponentially with the system size.  
Consequently, we argue that practical computational approaches towards the sign problem, rather than focusing on \emph{exactly curing} it, should target the less ambitious yet practically meaningful task of approximately solving or \emph{easing} it in the best possible way. 

Here, we propose a systematic, generally applicable, and practically feasible methodology for easing the sign problem via basis rotations that allows for a meaningful rigorous assessment of this task.
An appealing feature of our framework is that it neither requires any a priori knowledge about the physics of a problem
 nor depends on specifics of a given simulation procedure, in contrast to other known refinements of QMC. 
At the heart of our approach lies a formulation of the easing problem in terms of a simple, efficiently computable measure of approximate stoquasticity that generically quantifies the sampling complexity. 

The sample complexity of a QMC algorithm can be linked to the size of the inverse average sign, which directly bounds the variance of the QMC estimator~\cite{troyer_computational_2005}. 
In an attempt to ease the sign problem of a given Hamiltonian it is therefore natural to try and improve the average sign. 
For a few specific models such improvements have indeed been achieved by different means: 
for example, one can exploit known physics to find bases with improved average sign~\cite{shinaoka_negative_2015,wessel_efficient_2017} that are often induced by sparse representations~\cite{mcclean_clock_2015,thomas_stochastic_2015,dobrautz_compact_2019}.
For particular observables, one can also exploit clever decompositions of the Monte Carlo estimator into clusters with non-negative sign~\cite{bietenholz_meron-cluster_1995,chandrasekharan_meron-cluster_1999,henelius_sign_2000,nyfeler_nested_2008,huffman_solution_2016,hann_solution_2017,hen_resolution_2019}. 

However, the sample complexity of computing the average sign via QMC is given by its very value and typically scales exponentially in the system size. 
Ironically, easing the sign problem by optimizing the average sign is therefore typically infeasible whenever there is a sign problem. 
One would hence like to quantify the severeness sign problem in terms of a quantity that is efficiently computable for physical Hamiltonians -- a crucial property to be practically useful in a general approach to easing the sign problem. 

Building on the notion of stoquasticity, for a real $D \times D$ Hamiltonian matrix $H$, we propose the sum of all non-stoquastic matrix entries
\begin{align}
\label{eq: non-stoquasticity measure}
	\nu_1(H) \coloneqq D^{-1}\norm{ \nonstoq H }_{\ell_1}  ,
\end{align}
as a natural measure of non-stoquasticity\ifarxiv~\footnote{
In contrast to Ref.~\cite{marvian_computational_2018} 
where \emph{term-wise stoquasticity} is considered, our definition remains on the level of the global Hamiltonian.
This is because positive matrix elements of the local Hamiltonian terms may cancel in the global matrix representation. 
Term-wise stoquasticity is thus meaningful (only) in the following sense: 
if a Hamiltonian admits a term-wise stoquastic local basis, it has no sign problem on an \emph{arbitrary} lattice. 
However, for any given graph, it might well be possible to fully cure the sign problem using an allowed set of transformations even if the Hamiltonian \emph{cannot} be made term-wise stoquastic. 
}\fi\ in order to quantify the sampling complexity of a QMC algorithm in generic instances. 
Here, as throughout this work, we denote the non-stoquastic part of the Hamiltonian by $\nonstoq H$ which is defined by $(\nonstoq H)_{i,j} = h_{i,j}$ for $h_{i,j} > 0 $ and $ i \neq j$, and zero otherwise. 
Moreover, $\norm{H}_{\ell_1} = \sum_{i,j} \abs{h_{i,j}}$ is the vector-$\ell_1$-norm. 

For local Hamiltonians on bounded-degree graphs such as regular lattices this measure can be efficiently computed from the non-stoquastic entries of the local terms themselves -- for translation-invariant Hamiltonians even with constant effort. 
But we can also go beyond that and prove that, for $2$-local Hamiltonians acting on any graph, the measure $\nu_1$ can be efficiently approximated up to any inverse polynomial error\ifarxiv; see Theorem~\ref{thm:approximating non-stoq arb graph}\fi\ifjournal; see Sec.~\ref{sec:computing non-stoq} of the Supplementary Material for details\fi. 
This result renders our measure applicable to problems with long-range and low-degree interactions as they arise, for example, in quantum chemistry.

In principle, one can also conceive of other measures of non-stoquasticity such as the $\ell_{1 \rightarrow 1}$-norm or the $\ell_2$-norm of the non-stoquastic part of $H$. 
We argue that the $\ell_1$-norm is the most meaningful measure that is agnostic to any particular structure of the Hamiltonian matrix and therefore the most versatile measure for a general approach to easing the sign problem.
What is more, it acts as a natural regularizer promoting a sparse representation~\cite{foucart_mathematical_2013} in the spirit of Refs.~\cite{mcclean_clock_2015,thomas_stochastic_2015,dobrautz_compact_2019}.

But how does the non-stoquasticity relate to the sample complexity of a QMC simulation?  
We find that it is in fact impossible to directly connect a continuous measure of non-stoquasticity to the average sign, which takes on its maximal value at unity and achieves this value for stoquastic Hamiltonians: 
We can construct exotic examples of highly non-stoquastic Hamiltonians with large positive off-diagonal entries which also have unit average sign. 
Conversely, we provide an example of a Hamiltonian with arbitrarily small non-stoquasticity for which the average sign nearly vanishes.

On the one hand, our examples demonstrate a high sensitivity of the average sign to the Monte Carlo parameters.  
On the other hand, they also require a malicious interplay between the Hamiltonian matrix entries and highly fine-tuned Monte Carlo parameters. 
We therefore expect that, in generic situations, the non-stoquasticity measure $\nu_1$ meaningfully quantifies the sample complexity of QMC.
We give analytical arguments that this is actually the case and numerically find that the average sign of generic two-local Hamiltonians scales exponentially in $\nu_1$; see Sec.~\ref{sec:av sign vs nonstoq}\ifjournal\, of the Supplementary Material~\cite{suppmaterial} for details\fi. 
Thus, we provide evidence that the non-stoquasticity of a local Hamiltonian meaningfully quantifies its sign problem.

\subsection*{Easing in practice}

This leads us to the question:
Can we practically ease the sign problem of physical Hamiltonians by minimizing non-stoquasticity? 
To study this second question, we consider translation-invariant nearest-neighbour Hamiltonians in a quasi one-dimensional geometry \cite{mikeska_one-dimensional_2004}. 
Quasi one-dimensional systems, such as anti-ferromagnetic Heisenberg Hamiltonians on ladder geometries \cite{dagotto_surprises_1996,takano_spin_1996} are the simplest non-trivial systems that exhibit a sign problem since they admit the phenomenon of geometric frustration \cite{sandvik_computational_2010}. 
Frustration gives rise to a plethora of phenomena arising in quasi one-dimensional systems such as the emergence of quantum spin liquids \cite{meng_coupled-wire_2015,huang_coupled_2017} and the interplay of spin-$1/2$ and spin-$1$ physics \cite{nietner_composite_2017}. 
They are also somewhat more realistic descriptions of actual low-dimensional experimental situations than simple one-dimensional chains, serving as a model for small couplings in the transverse direction~\cite{takano_spin_1996,yoshida_spin-disordered_2015,lai_coexistence_2017}. 
Therefore quasi one-dimensional systems are often seen as a stepping stone towards studying higher dimensions~\cite{iglovikov_geometry_2015}, where the sign problem inhibits QMC simulations~\cite{carrasquilla_two-dimensional_2015}, and thus serve as the perfect playground for a proof of principle.

\begin{figure*}[t]
	\includegraphics[width = \textwidth]{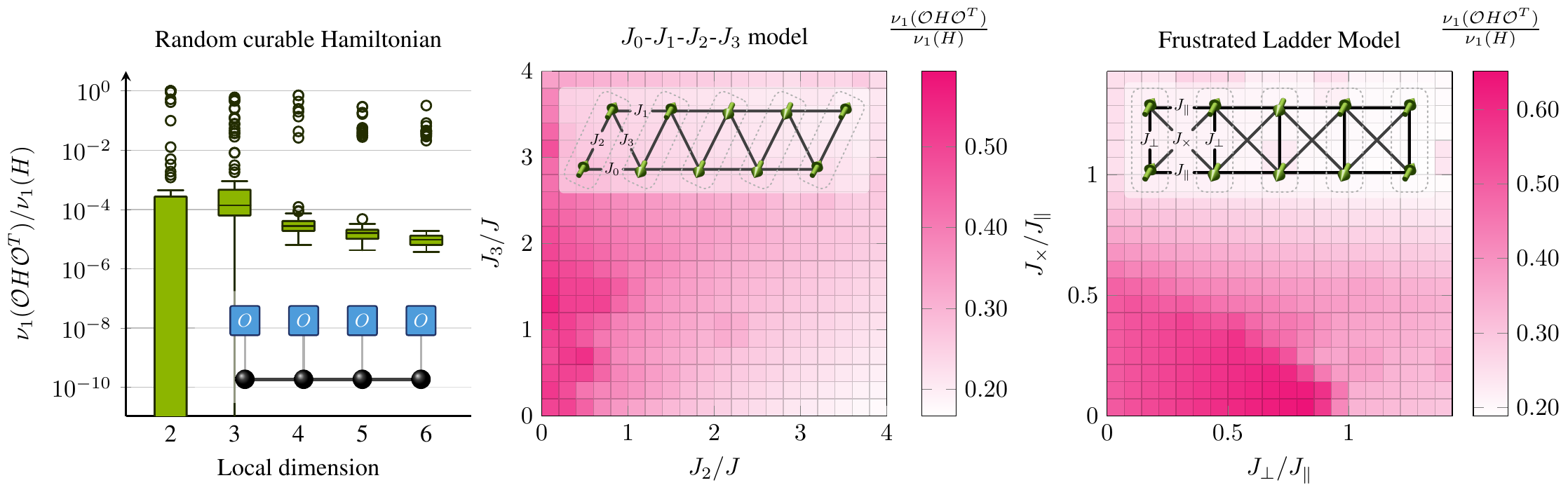}
	\caption{ We optimize the non-stoquasticity $\nu_1$ of translation-invariant, two-local Hamiltonians over on-site orthogonal transformations $\mc O = O^ {\otimes n} $ using a conjugate gradient method for manifold optimization \cite{abrudan_conjugate_2009,optimization_numerics}. 
	Figure
	\emph{(a)} shows the relative non-stoquasticity improvement of random two-local Hamiltonians that are known to admit an on-site stoquastic basis. 
	For each local dimension $100$ instances are drawn and the results displayed as a box plot according to Ref.~\cite[2.16]{boxplot_norm}, where whiskers are placed at $1.5$ times the interquartile range and circles denote outliers. 
	This serves as a benchmark of our algorithm, which for almost all instances accurately recovers a stoquastic on-site basis. 
	Figure
	\emph{(b)} displays the optimized non-stoquasticity of the anti-ferromagnetic \jmodel-Heisenberg model relative to the computational basis as a function of $J_2/J,J_3/J$, where $J_0 = J_1 = J$. 
	The algorithm is initialized in a Haar random orthogonal on-site basis. 
	This model is known to admit a \emph{non-local} stoquastic basis for $J_3 \geq J_0 + J_1$ \cite{nakamura_vanishing_1998}. Figure
	\emph{(c)} shows the optimized non-stoquasticity of the anti-ferromagnetic Heisenberg ladder illustrated in the inset with couplings $J_\parallel, J_\perp, J_\times$ relative to the computational basis as a function of $J_\perp/J_\parallel$ and $J_\times /J_\parallel$.
	We initialized the algorithm at the identity matrix (that was randomly perturbed by a small amount). 
	The phase diagram of the non-stoquasticity qualitatively agrees with the findings of Ref.~\cite{wessel_efficient_2017}, where the stochastic series expansion (SSE) QMC method was studied. 
	There, it was found that the sign problem can be completely eliminated for a completely frustrated arrangement where $J_\times = J_\parallel$, while the sign problem remains present for partially frustrated couplings $J_\times \neq J_\parallel$. 
	However, throughout the parameter regime the stoquasticity remains non-trivial, which may be due to the fact that the optimization algorithm converges to local minima.
	\label{fig:stoquastic plots}
	}
\end{figure*}

As a meaningful simple ansatz class, we consider on-site orthogonal transformations $O \in O(d)$ of the type 
\begin{align}
\label{eq:local basis choice main text}
	H = \sum_{i = 1}^{n} T_i(h) \mapsto O^{\otimes n} H (O^T)^{\otimes n} ,
\end{align}
for Hamiltonians $H$ acting on $n$ qudits with local dimension $d$. 
Here, $T_i(h) $ denotes the translation of a two-local term $h$ to site $i$. On-site transformations can be handled particularly well as they preserve locality and translation-invariance of local Hamiltonians. 
In particular, for such transformations, the global non-stoquasticity measure can be expressed locally in terms of the transformed term $O^{\otimes 2 } h (O^T)^{\otimes 2} $ so that the optimization problem has constant complexity in the system size. 
This constitutes an exponential improvement over approaches that directly optimize the average sign.  

To optimize the non-stoquasticity in this setting, we have implemented a geometric optimization method suitable for group manifolds, namely, a conjugate gradient descent algorithm over the orthogonal group $O(d)$ \cite{abrudan_conjugate_2009,optimization_numerics}\ifjournal; see Sec.~\ref{sec:practical easing} of the Supplementary Material~\cite{suppmaterial} for details\fi. 
In Fig.~\ref{fig:stoquastic plots}(a) we show that, generically, the algorithm accurately recovers an on-site stoquastic basis for random Hamiltonians which are known to admit such a basis \emph{a priori}. 
This shows that the heuristic algorithm successfully minimizes the non-stoquasticity and thus serves as a benchmark for its functioning. 

We now apply the algorithm to frustrated anti-ferromagnetic Heisenberg Hamiltonians on different ladder geometries; see Fig.~\ref{fig:stoquastic plots}(b) and (c). 
Ladder geometries are not only interesting for the reasons described above, but also because in spite of frustration effects they often admit sign-problem free QMC methods \cite{nakamura_vanishing_1998,honecker_thermodynamic_2016,wessel_efficient_2017}. 
For both the \jmodel-model studied in Ref.~\cite{nakamura_vanishing_1998} and the frustrated Heisenberg ladder studied in Refs.~\cite{honecker_thermodynamic_2016,wessel_efficient_2017}, we find a rich optimization landscape in which a relative improvement of the non-stoquasticity by a factor of $2$ to $5$ can be achieved depending on the region in the phase diagram. 
Importantly and in spite of those seemingly moderate improvements of non-stoquasticity, we find that the sample complexity of QMC as governed by the inverse average sign is greatly diminished to approximate unity in large regions of the parameter space for the frustrated ladder model; see Fig.~\ref{fig:frustrated ladder av sign small}. 

\begin{figure}
	\includegraphics{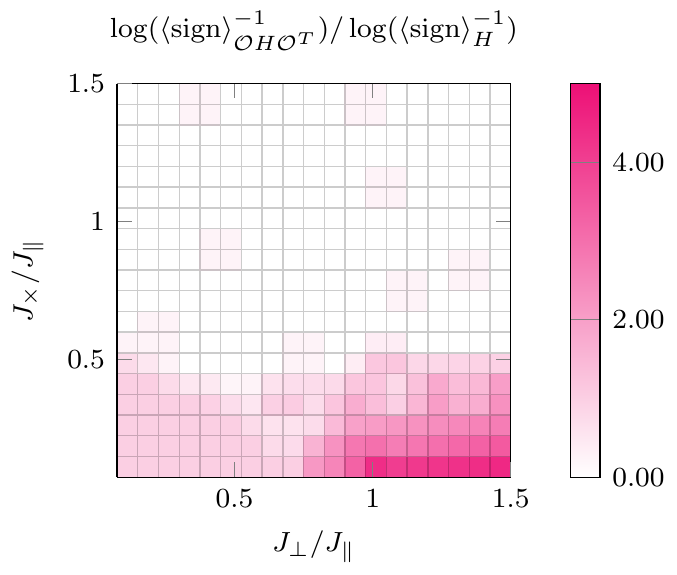}
	\caption{
	Improvement of the inverse average sign $\l \sign \r^{-1}$ concomitant with the improvement in non-stoquasticity of Fig.~\ref{fig:stoquastic plots}(c) for the frustrated ladder model as measured by the ratio of its logarithm before optimization compared to that after optimization. 
  	We compute the average sign via exact diagonalization for a ladder of $2 \times 4$-sites, $m = 100$ Monte Carlo steps and inverse temperature $\beta = 1$.
 }
	\label{fig:frustrated ladder av sign small}
\end{figure}

It may well be the case that no stoquastic dimer basis exists \emph{even though} other variants of QMC do not incur a sign problem for such basis choices:  
in Ref.~\cite{nakamura_vanishing_1998} a stoquastic but \emph{non-local} basis of the \jmodel-model is identified for values of $J_2 \geq J_0 + J_1$, indicating that more general ansatz classes may well help to further improve the non-stoquasticity. 
We also observe that first-order optimization algorithms such as the employed conjugate gradient method encounter obstacles due to the rugged non-stoquasticity landscape. 
Intuitively, this landscape is governed by the combinatorial increase of possible assignments of signs to the Hamiltonian matrix elements. 

The findings of our proof-of-principle study are twofold: 
on the one hand, they show that one can in fact efficiently optimize the non-stoquasticity for translation-invariant problems that admit a stoquastic basis lying within the ansatz orbit.
They also further substantiate the claim that optimizing non-stoquasticity typically eases the sign problem and dampens the increase in sampling complexity.
What is more, they indicate that more general ansatz classes such as quasi-local circuits yield the promise to further reduce the non-stoquasticity of ladder models.
We therefore expect that optimizing non-stoquasticity is a feasible and promising means to reduce the sign problem for many different systems, including two-dimensional lattice systems, by exploiting the flexibility offered by larger ansatz classes within our framework. 
On the other hand, already in our small study we encountered obstacles preventing efficient optimization of the non-stoquasticity in the guise of a complicated and rugged optimization landscape.

\begin{figure*}[t]
	\centering
	\includegraphics{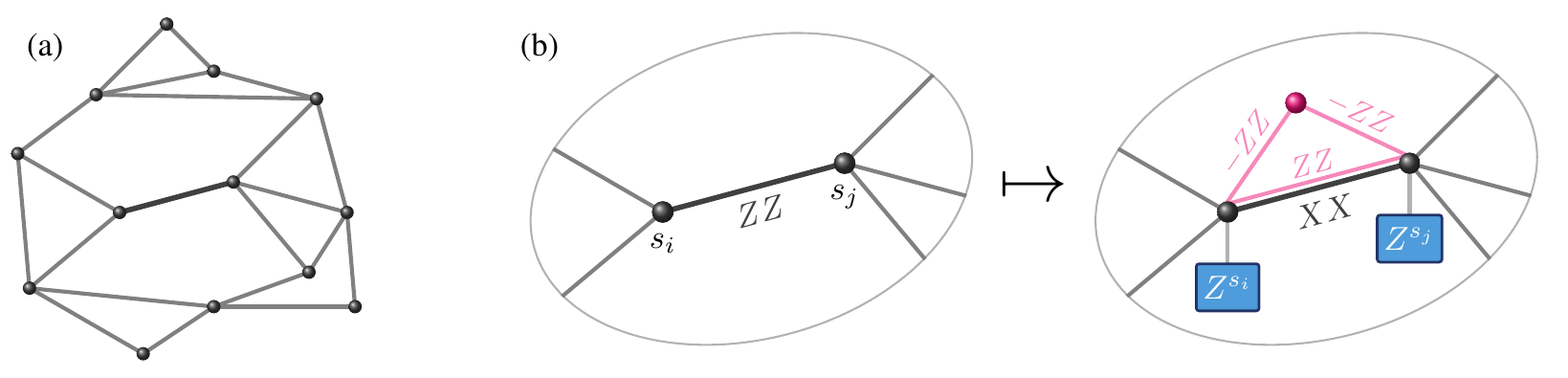}
	\caption{
	\label{fig:maxcut embedding}
	Constructing a Hamiltonian whose sign problem is \np-hard to ease under orthogonal on-site transformations.
	\emph{(a)} To prove \np-completeness of \se, we reduce it to the \maxcut-problem which asks for the ground-state energy of an anti-ferromagnetic Ising Hamiltonian $H$ on a graph $G$. 
	\emph{(b)} 
	In our encoding, we map $H$ to a Hamiltonian $H'$ in which all $ZZ$-interactions are replaced by $XX$-interactions and translate the spin 
	configurations $(s_1, \ldots, s_n) \in \{0,1\}^n$ of the anti-ferromagnetic Ising model to on-site transformations $Z_1^{s_1} \cdots Z_n^{s_n}$. 
	To achieve this restriction, we penalize all other transformations by adding an ancilla qubit $\xi_{i,j}$ for every edge $(i,j)$ of $G$ and adding the interaction term $C(Z_i Z_j - Z_i Z_{\xi_{i,j}} - Z_j Z_{\xi_{i,j}})$ with a suitably chosen constant $C > 0$. 
	We obtain that $\nu_1(H')$ can be eased below a certain value if and only if the ground state energy of $H$ is below that value to begin with, thus establishing the reduction. 
	}
\end{figure*}

\subsection*{The computational complexity of \se}

Fundamentally, our findings thus raise the third question: 
How far can an approach to easing the sign problem using optimization over local bases carry in principle? 
In our main complexity-theoretic result, we systematically study the fundamental limits of minimizing non-stoquasticity as a means to ease the sign problem.
To do so, we complement the pragmatic mindset of this work with the rigorous machinery of computational complexity theory, asking the question: 
What is the computational complexity of optimally easing the sign problem?
In order to formalize this question, 
we introduce the corresponding decision problem: 
\begin{definition}[\se]
	Given an $n$-qubit Hamiltonian $H$, constants $B > A \geq 0$ with $B - A \geq 1/\poly(n)$, and a set of allowed unitary transformations $\mc U$, decide which of the following is the case: 
	\begin{align}
		\text{YES}:& \quad \exists U \in \mc U: \nu_1(U H U^\dagger) \leq A, \text{ or } \label{eq:se yes}\\
		\text{NO}:&  \quad\forall U \in \mc U: \nu_1(U H U^\dagger) \geq B. \label{eq:se no}
	\end{align}
\end{definition}

We derive the computational complexity of the sign easing problem in simple settings, namely for $2$-local Hamiltonians, allowing for on-site orthogonal Clifford operations as well as for on-site general orthogonal transformations. 
We prove that under both classes of transformations \se\ is \np-complete. 
Intriguingly, this holds true even in cases in which the curing problem can be decided efficiently, namely, for strictly $2$-local XYZ Hamiltonians of the type considered in Refs.~\cite{klassen_two-local_2018,klassen_hardness_2019}. 

\begin{theorem}[Complexity of \se]
\label{thm:complexity of se}
	\se\, is \np-complete for $2$-local (XYZ) Hamiltonians under 
	\begin{enumerate}[i.]
		\item \label{item:se i} on-site orthogonal Clifford 
		transformations, and 
		\item \label{item:se ii}on-site general orthogonal 
		transformations. 
	\end{enumerate}
\end{theorem}

From a practical perspective, our results pose limitations on the worst-case runtime of algorithms designed to find optimal QMC bases for the physically relevant case of $2$-local Hamiltonians. 
From a complexity-theoretic perspective, they manifest a sign problem variant of the dichotomy between the efficiently solvable \twosat-problem to decide whether there exists a satisfying assignment for a $2$-local sentence, and the \np-complete \maxtwosat-problem asking what is the least possible number of broken clauses. 
They thus complete the picture drawn by Refs.~\cite{klassen_two-local_2018,marvian_computational_2018,klassen_hardness_2019} regarding the connection between satisfiability problems and the problems of curing and easing the sign problem on arbitrary graphs, a state of affairs which we illustrate in Table~\ref{tab:2sat-max2sat}. 
It is natural to ask the question how far this connection extends and what we can learn from it about efficiently solvable instances. 
For example, one may ask, whether results about the hard regions of \threesat\ and \maxtwosat\ carry over to the problems of curing and easing the sign problem.

\begin{table}[b]
\centering
	\begin{tabular}{c >{\ }c<{\quad}  c  >{\quad}c }
	\toprule
		Satisfiability & Stoquasticity & Complexity &  Refs.\\ \midrule
		\threesat &Curing $2$+$1$-local $H$ &  \np-complete & \cite{marvian_computational_2018,klassen_hardness_2019}\\
		 \twosat & Curing strictly $2$-local $H$& in \P & \cite{klassen_two-local_2018,klassen_hardness_2019}\\
		 \maxtwosat  & Easing strictly $2$-local $H$& \np-complete & \emph{here} \\ \bottomrule
	\end{tabular}
	\caption{
	The satisfiability equivalent of curing the sign problem is to decide whether a given sentence is satisfiable, while the equivalent of easing is to find the minimal number of clauses that are violated by a sentence.
	Similarly, results on the computational complexity of curing and easing the non-stoquasticity of a local Hamiltonian $H$ are in one-to-one correspondence with the hardness of satisfiability problems. 
	\label{tab:2sat-max2sat}
	}
\end{table}

\ifjournal
\section*{Materials and Methods}
\fi

We prove Theorem~\ref{thm:complexity of se}\,\ref{item:se i} and \ref{item:se ii} as Theorems \ref{thm:clifford hardness} and \ref{thm:orthogonal hardness}\ifjournal\, in the Supplementary Material~\cite{suppmaterial}\fi. 
The essential idea of our proof, sketched below and illustrated in Fig.~\ref{fig:maxcut embedding}, is to design a corresponding Hamiltonian such that if the sign problem could be optimally eased for this Hamiltonian under the respective ansatz class, one could also find the ground state energy of the original anti-ferromagnetic Ising Hamiltonian, a task that is \np-hard to begin with. 
It is straightforward to prove versions of Theorem~\ref{thm:complexity of se} for any $\ell_p$-norm of the non-stoquastic part of $H$ with finite $p$ as a measure of non-stoquasticity. 
Our result is therefore independent of the particular choice of ($\ell_p$) non-stoquasticity measure. 

\begin{proof}[Proof sketch]
\se\ for arbitrary $2$-local Hamiltonians is contained in \np\ -- given a basis transformation, we can approximate the measure of non-stoquasticity from the transformed local terms up to any inverse polynomial error and hence verify the YES-case~\eqref{eq:se yes}; see Theorem~\ref{thm:approximating non-stoq arb graph}\ifjournal \, in the Supplementary Material~\cite{suppmaterial}\fi. 

The key idea of the harder direction of the proof is to encode the promise version of the \maxcut-problem into the \se-problem. 
An instance of \maxcut\ is given by a graph $G = (V,E)$, and the problem is to decide whether the ground-state energy of the anti-ferromagnetic (AF) Ising Hamiltonian 
\begin{align}
	H = \sum_{(i,j) \in E} Z_i Z_j \, ,
\end{align}
is below a constant $A$ or above $B$. 
Here, $Z_i$ is the Pauli-$Z$-operator acting on site $i$. 
We now define a Hamiltonian $H'$ in which we replace every $Z_iZ_j$ interaction of $H$ by an $X_i X_j$ interaction as we illustrate in Fig.~\ref{fig:maxcut embedding}. 
To understand our embedding, suppose that we perform basis changes only by applying $Z$ or $\id$ at every site. 
In this case a Hamiltonian term can be made stoquastic if and only if $X_i X_j \mapsto - X_i X_j$ which is achieved by a transformation $Z^{s_i} Z^{s_j}$ with $(s_i ,s_j)= (0,1) \vee (1,0) $. 
A term remains stoquastic for $(s_i, s_j) = (1,1) \vee (0,0)$.
This provides a direct mapping between spin configurations $(1,0)$ and $(0,1)$, which do not contribute to the ground state energy of the anti-ferromagnetic Ising model and transformations that make local terms in $H'$ stoquastic and thus decrease the non-stoquasticity. 

To prove the theorem for arbitrary on-site Clifford and orthogonal transformations, we introduce an additional qubit $\xi_{i,j}$ for every edge $(i,j)$ and add interaction terms $C(Z_i Z_j - Z_i Z_{\xi_{i,j}} - Z_j Z_{\xi_{i,j}})$ to $H'$ with constant $C = 2 \deg(G)$, where $\deg(G)$ is the degree of the interaction graph $G$, see Fig.~\ref{fig:maxcut embedding}(b). 
These terms penalize all other transformations such that the optimal non-stoquasticity of $H'$ is always achieved for transformations of the form $Z_1^{s_1} \cdots Z_n^{s_n}$ with $(s_1, \ldots, s_n) \in \{ 0,1\}^n$. 
For example, suppose that we apply Hadamard transformations to all sites $i,j,\xi_{i,j}$, then the $ZZ$ interactions and $XX$ interactions change roles so that the non-stoquasticity cannot be decreased by such a transformation. 
Showing this for all possible transformations constitutes the main technical part of the proof. 
\end{proof}

Since \maxcut\, is a variant of the \maxtwosat-problem our results not only manifest but also crucially utilise the \twosat-\maxtwosat\, dichotomy.
Notice that since the \maxcut-problem is \np-hard already on 
subgraphs of the double-layered square lattice \cite{barahona_computational_1982}, which has degree six, hard instances of the sign-easing problem occur already for low-dimensional lattices with small (constant) interaction strength.

In our complexity-theoretic analysis, we have focused on the computational complexity of easing the sign problem as the size of an arbitrary input graph is scaled up, in the same mindset as Refs.~\cite{troyer_computational_2005,klassen_two-local_2018,marvian_computational_2018,klassen_hardness_2019}. 
We expect, however, that the complexity of \se\ scales similarly in the size of the lattice unit cell and the local dimension of translation-invariant systems such as those discussed above. 

\ifjournal
\section*{Discussion}
\fi
\subsection*{Summary}

Let us summarize: 
Our work introduces the sign easing methodology as a systematic novel paradigm useful for assessing and understanding the sign problem of QMC simulations. 
We ask and answer three central questions using complementary methods from theoretical and applied computer science as well as 
from physics. 
First, we define a measure of non-stoquasticity suitable for easing the sign problem and extensively discussed its relation to the average sign. 
Second, we demonstrate that one can feasibly optimize this measure over local bases in simple settings by applying geometric optimization methods. 
Finally, we establish the computational complexity of sign easing in a broader but still simple setting.
In this way, our work not only identifies a means of easing the sign problem and demonstrates its feasibility and potential, but also shows up its fundamental limitations in terms of computational complexity. 
Even more so, we are confident that the framework of our work provides both valuable guidance and the practical means for future research on systematically easing the sign problem of Hamiltonians that are particularly interesting and relevant in condensed-matter and material science applications. 

\subsection*{Outlook}

As a first general and systematic attempt to easing the sign problem,
we have restricted the focus of this work in several ways. 
As such, a number of questions, generalizing our results in different directions, are left open. 

First, we have restricted our discussion to the prominent world-line Monte Carlo method to maintain clarity throughout the manuscript. 
We are confident, however, that our results find immediate application for other Monte Carlo methods such as stochastic series expansion Monte Carlo and determinantal Monte Carlo \cite{landau_guide_2000,sandvik_computational_2010} as well as diffusion Monte Carlo techniques such as full-configuration-interaction Monte Carlo~\cite{booth_fermion_2009}. 
Similar sign problems involving the sampling from quasi-probability distributions also appear in different contexts, for example, in 
approaches to the classical simulation of quantum circuits~\cite{dawson_quantum_2004,jordan_quantum-merlin-arthurcomplete_2010,pashayan_estimating_2015} or high-energy physics~\cite{PhysRevD.66.106008}.
In these contexts, too, the problem of finding better bases in which to perform the sampling appears. 
While the framework developed in this work uses the specific features of QMC, the general idea and mindset behind it applies to all basis-dependent sign problems. Our work thus paves the way towards easing sign problems in a 
plethora of contexts. 

Second, we have only considered real-valued Hamiltonians and transformations which preserve this property. 
For general complex-valued Hamiltonians, the sign problem takes the form of a \emph{complex phase problem}. 
A natural follow-up of our work is to explore how our results on easing the sign problem generalize to the complex phase problem.

Third, we have put an emphasis on the conjugation of 
Hamiltonians under on-site
Clifford and orthogonal circuits. 
In principle, one may also allow for arbitrary quasi-local
circuits, as long as the conjugation can be efficiently computed; albeit of exponentially increasing effort with the support of the involved unitaries. 
This leads to the interesting insight that within the trivial phase of matter, one can always remove the sign problem: 
One has to conjugate the Hamiltonian with the quasi-local
unitary that brings a given Hamiltonian into an on-site form of a \emph{fixed point Hamiltonian}.
For given Hamiltonians, this may be impractical, of course. 
In this sense, one can identify trivial quantum phases of matter as \emph{efficiently computable phases of matter}, an intriguing state of affairs from a conceptual perspective. 
Conversely, for topologically
ordered systems, there may be topological obstructions to curing the sign problem by any quasi-local circuit~\cite{hastings_how_2015,ringel_quantized_2017}, giving rise to an entire phase of matter that exhibits an intrinsic sign problem. 
For example, the fixed point Hamiltonians
of the most general class of non-chiral topologically ordered systems, the Levin-Wen models
\cite{PhysRevB.71.045110}, are associated with $12$-local Hamiltonians, many of which are expected to not be
curable from their sign problem. 
This insight further motivates to study the sign easing problem for efficiently computable subgroups of local unitaries from a perspective of topological phases of matter.

Our work also opens up several paths for future research. 
The immediate and practically most relevant direction is of course to find the best possible way of minimizing the non-stoquasticity of translation-invariant systems and to explore how well the sign problem can be eased in systems that are not yet amenable to QMC. 
We have already introduce a flexible optimization approach which can be straightforwardly applied to a wide range of translation-invariant systems and ansatz classes in any dimensionality.
In this respect, it will be interesting to compare possible ways of optimizing the sign problem via different measures~\cite{levy_mitigating_2019} and optimization algorithms~\cite{torlai_wavefunction_2019} in various systems \cite{kim_alleviating_2019}.

Furthermore, in our hardness proof we have shown that the easing problem is intricately related to satisfiability problems. 
Building on this connection, an exciting direction of research is to combine highly efficient \sat-solvers that are capable of exploring 
combinatorically large sets, with manifold optimization techniques that are able to handle rich geometrical structures, in the spirit of recent work~\cite{shoukry_smc:_2017}. 
While our hardness result shows up fundamental limitations of \se\ in the general case, it thus also opens the door to potentially solve the sign easing problem in relevant instances by applying methods well known in computer science to relaxed versions of the easing problem. 
One may thus hope that for large classes of relevant instances for which minimizing non-stoquasticity is actually tractable.

A question closely related to the sign easing problem is the following: 
How hard is it to find the ground state energy of a stoquastic Hamiltonian -- a sub-problem of the so-called local Hamiltonian problem. 
The computational complexity of this 
\emph{stoquastic local Hamiltonian problem} poses fundamental 
limitations on the classical simulatability of Hamiltonians which 
do not suffer from a sign problem and are therefore amenable to QMC simulations. 
It has been shown that the $2$-local stoquastic Hamiltonian problem is complete for the class {\stoqma} \cite{bravyi_complexity_2006,bravyi_complexity_2009}, a class intermediate between {\am} and {\ma} that also functions as a genuinely intermediate class in the complexity classification of local Hamiltonian problems \cite{cubitt_complexity_2016}, even when extending to the full low-energy spectrum \cite{cubitt_universal_2018}. 
The results of Ref.~\cite{bravyi_complexity_2006} also imply that we cannot expect to efficiently find a stoquastic local basis for arbitrary local Hamiltonians unless the unlikely complexity-theoretic equality $\am = \qma$ holds.

Indeed, for efficiently curable Hamiltonians, the local Hamiltonian problem is reduced to a stoquastic local Hamiltonian problem. 
Conversely, both the easing problem and the stoquastic local Hamiltonian problem contribute to the hardness of a QMC procedure. 
For a given Hamiltonian, QMC may thus be computationally intractable for two reasons: 
it is hard to find a basis in which the Hamiltonian is stoquastic, or cooling to its ground state is computationally hard in its own right.
In a QMC algorithm, the latter hardness is manifested as a Markov chain Monte Carlo algorithm not converging in polynomial time. 
This may be the case even for classical models such as Ising spin glasses \cite{barahona_computational_1982}. 

An important open question is how the hardness of easing the sign problem and the hardness of sampling from the estimator distribution are related in specific cases.
For example, when improving the average sign, the hardness of a problem that was manifest in an increased sample complexity of the Monte Carlo estimator, might be `transferred' to the hardness of sampling from the resulting distribution. 
On the other hand, there might be instances in which the only obstacle in the way of an efficient simulation is to find a certain basis in which the corresponding Hamiltonian has a large average sign, but, given that basis, QMC runs efficiently.

\ifarxiv
\subsection*{Overview}

The plan for the technical part of this work is as follows: 
In Section~\ref{sec:sign problem exposition} we sketch the idea of world-line QMC methods and explain how the sign problem arises there. 
In Section~\ref{sec:av sign vs nonstoq} we then discuss the relation between the average sign and non-stoquasticity. 
There, we construct examples showing that the two are in general unrelated (\ref{sec:examples}), but then continue to argue both analytically (\ref{sec:motivating non-stoquasticity}) and numerically (\ref{sec:av sign numerics}) that the non-stoquasticity $\nu_1$ defined in Eq.~\eqref{eq: non-stoquasticity measure} is a meaningful and efficiently computable (\ref{sec:computing non-stoq}) measure of the sign problem. 
In Section~\ref{sec:practical easing} we perform a proof-of-principle numerical study showing that easing is both feasible and meaningful for translationally invariant models with a sign problem. 
In Section~\ref{sec:complexity} we then study the fundamental limitations of a systematic approach to the sign problem in proving the computational hardness of \se\ when allowing for both orthogonal Clifford (Theorem~\ref{thm:clifford hardness}) and general orthogonal transformations (Theorem~\ref{thm:orthogonal hardness}). 

\fi

\ifjournal
	\paragraph*{Acknowledgements}
We are immensely grateful for the many fruitful discussions which have helped shape this work -- with Albert 
Werner, Martin Schwarz, Juani Bermejo-Vega, 
and Christian Krumnow in early stages of the project; more recently with Matthias Troyer, Joel Klassen, Marios Ioannou, Maria Laura Baez, Hakop Pashayan, Simon Trebst, Augustine Kshetrimayum, Alexander Studt and Alex Nietner. 
We also thank Barbara Terhal, Marios Ioannou, Jarrod McClean and Maria Laura Baez for helpful comments on our draft. 
D.\,H., I.\,R.\ and J.\,E.\ acknowledge financial support 
from the ERC (TAQ), 
the Templeton Foundation, the DFG 
(EI 519/14-1, EI 519/9-1, EI 519/7-1, CRC 183 in project B01), 
and the European Union's Horizon 2020 research and innovation programme under grant agreement No 817482 (PASQuanS). D.\,N.\ 
has received funding from the People Programme (Marie Curie Actions) EU’s 7th Framework Programme under REA grant agreement No. 609427. His research has been further co-funded by the Slovak Academy of Sciences, 
as well as by the Slovak
Research and Development Agency grant QETWORK APVV-14-0878 and VEGA MAXAP 2/0173/17.


	\paragraph*{Author contributions}
	D.\,H.\ and I.\,R.\ conceived of the non-stoquasticity measure and its relation to the average sign, and carried out all analytical and numerical calculations.
	D.\,H.\ and D.\,N.\ conceived of the complexity proof idea. 
	J.\,E.\ contributed to all aspects of this work.
	All authors participated in discussions and contributed to writing the manuscript.

	\paragraph*{Data Availability.} The python package designed for the numerical simulations is available publicly at~\cite{optimization_numerics}, making the results shown reproducible. 
	All analytical calculations, in particular the explicit proof of Theorem~\ref{thm:complexity of se} are presented in \detailsandproof.

	\paragraph*{Competing interests} The authors declare no competing interests.

	\putbib
	\end{bibunit}


	\begin{bibunit}
	\onecolumngrid
	\cleardoublepage
	\setcounter{page}{0}
	\setcounter{equation}{0}
	\setcounter{footnote}{0}
	\setcounter{figure}{0}
	\thispagestyle{empty}
	\begin{center}
	\textbf{\large Supplementary Material for ``Easing the Monte Carlo Sign Problem''}\\
	\vspace{2ex}
	Dominik Hangleiter, Ingo Roth, Daniel Nagaj, and Jens Eisert
	\vspace{2ex}
	\end{center}

	\twocolumngrid
	\renewcommand\thesection{S\arabic{section}}
	\renewcommand\thefigure{S\arabic{figure}}
\fi



\section{The sign problem of Quantum Monte Carlo}
\label{sec:sign problem exposition}

We begin the technical part of this work with an exposition of the basics of Quantum Monte Carlo methods.
For the purpose of this work, we focus on the prominent \emph{world-line Monte Carlo} method of calculating partition functions and thermal expectation values of a Hamiltonian $H$ at inverse temperature $\beta$ \cite{landau_guide_2000}.
Here, both quantities are expressed as 
\begin{align}
	\label{eq:partition function monte carlo}
	Z_{\beta,H}  & \simeq \tr [ T_m^m] = \sum_{\vec \lambda \in \Lambda_{m+1}, \, \lambda_{m+1} = \lambda_1}a (\vec \lambda) \\ 
	\label{eq:observable monte carlo}
	\l O \r _{\beta,H}&  \simeq \frac 1 {Z_{\beta,H}}\tr [ T_m^m O ] =  \frac 1 {Z_{\beta,H}} \sum_{\vec \lambda \in \Lambda_{m+1}}a (\vec \lambda) O(\lambda_m|\lambda_1) ,
\end{align}
for large enough $m \in \mb N $ \emph{Monte Carlo steps} in terms of the amplitudes 
\begin{align}
	a(\vec \lambda) = T_m(\lambda_1| \lambda_2)  T_m(\lambda_2| \lambda_3)\cdots  T_m(\lambda_{m} | \lambda_{m+1}), 
\end{align}
on the configuration space $\Lambda_{m+1} = [\dim \mc H]^{\times ( m+1 )}$.
Here, we have defined the transfer matrix $ T_m(\lambda'|\lambda) = \bra{\lambda'} \id - \beta H/m \ket{\lambda}$ and in general denote the entries of a matrix $A$ as $A(\lambda_1|\lambda_2) = \bra{\lambda_1} A\ket{\lambda_2}$.
The computation of the partition function involves a summation over all closed paths of length $m$ (i.e., paths with periodic boundary conditions); the computation of general observables involves a summation over all open paths. 

For non-negative path weights, both quantities may be rewritten as expectation values in a probability distribution $q(\vec \lambda) = a(\vec \lambda)/\sum_{\vec \lambda} a(\vec \lambda)$, which reduces to $q(\vec \lambda) = a(\vec \lambda)/Z_{\beta,H}$ when computing the expectation value of diagonal observables. 
The sign problem is manifested in the fact that the off-diagonal entries of $H$ may be positive potentially implying that $a(\vec \lambda) < 0$. Therefore $q(\vec \lambda)$ is in general a quasi-probability distribution. 

To compute the quantities \eqref{eq:partition function monte carlo} and \eqref{eq:observable monte carlo} via Monte Carlo sampling, one constructs a linear estimator as the expectation value $\langle f \rangle_p = \sum_{\vec \lambda} p(\vec \lambda)  f(\vec \lambda)$ of a random variable $f$ distributed according to a probability distribution $p$.
By Chebyshev's inequality the statistical error $\epsilon$, i.e.\ the deviation from the mean, when averaging $s$ samples of an i.i.d.\ random variable $X$ is upper bounded by its variance 
\begin{align}
	\label{eq: error chebyshev}
	\epsilon \leq \sqrt{\var(X)/(s(1-\delta))}  \, , 
\end{align}
with probability at least $1- \delta$. 
Hence, to achieve any relative error $\tilde \epsilon$, the number of samples needs to grow with the variance of the random variable normalized by its expectation value. 
In fact, it can be easily shown that the variance-optimal estimator for the partition function $Z_{\beta, H}$ is given by the probability distribution $p(\vec \lambda) = |a(\vec \lambda)| /\norm{a}_{\ell_1}$ with $\norm{a}_{\ell_1} = \sum_{\vec \lambda} |a(\vec \lambda) |$ and the estimator $f(\vec \lambda) = \sign(a(\vec \lambda)) \cdot \norm{a}_{\ell_1}$ \cite{pashayan_estimating_2015}. 
The variance of this estimator is given by 
\begin{align}
	 \var_p(f)    & =  \norm{a}_{\ell_1}^2 ( \norm{q}_{\ell_1}^{2} - 1)  
\end{align}
and hence the relative error of the approximation by 
\begin{align}
\label{eq:relative error average sign}
	 \frac{\var_p(f)}{\l f \r_p^2 }& 
	 = \norm{q}_{\ell_1}^2 - 1 \equiv \l \sign \r_p^{-2} - 1 , 
\end{align}
where $\l \sign \r_p = 1/ \norm{q}_{\ell_1}$ is called the \emph{average sign} 
of the quasi-probability distribution $q$. 
One may interpret the average sign as the ratio between the partition functions of the original system with Hamiltonian $H$ acting on $n$ qubits and a corresponding `bosonic system' with Hamiltonian $H' = \stoqabs H$ as $\l \sign \r_p = \tr [\ee^{-\beta H}]/\tr [\ee^{-\beta H'}]$. 
Generically, such a quantity is expected to scale as $\ee^{- \beta n \Delta f}$, that is, inverse exponentially in the particle number $n$, the inverse temperature $\beta$, and the free energy density difference $\Delta f = f' - f \geq 0 $ between `bosonic' and original system~\cite{troyer_computational_2005}. 

In order to minimize the relative approximation error of a QMC algorithm, we therefore need to minimize the inverse average sign, or equivalently $\norm{q}_{\ell_1}$, over the allowed set of basis choices which we denote by $\mc U$. 
To optimally ease the sign problem in terms of its sample (and hence computational) complexity one therefore needs to solve the following minimization problem
\begin{align}
	\label{eq:relative error minimization problem}
	\min_{U \in \mc U } \norm{q}_{\ell_1}^2 - 1 = \min_{U \in \mc U } \frac{\tr[|U T_m U^\dagger|^m]^2}{\tr [T_m^m]^2}  -1, 
\end{align}
where as throughout this work $\abs{\cdot} $ denotes taking the entry-wise absolute value and not the matrix absolute value.


\section{The relation between the average sign and non-stoquasticity}
\label{sec:av sign vs nonstoq}

The difficulty in dealing with the minimization problem \eqref{eq:relative error minimization problem} is manifold. 
First, determining the quantity $\norm{q}_{\ell_1} = \tr[\abs{T_m}^m]/\tr[T_m^m]$ via QMC suffers from the very sign problem it quantifies: 
it can easily be checked that the relative variance of $\l \sign \r_p$ is precisely given by $\l \sign \r_p^{-2} -1$. 
It thus inherits the complexity of computing the partition function $Z_{\beta,H}$ in the first place. 
Na\"ive optimization of the term $\tr[\abs{T_m}^m]/\tr[T_m^m]$ even incurs the cost of diagonalizing the exponential-size matrices $T_m$ and $|T_m|$. 
Second, the optimization problem is non-convex and highly non-linear in the unitary transformation $ T \mapsto U T U ^\dagger$ with $U \in \mc U$.

While it might be possible to minimize the unitarily dependent term $\tr[\abs{T_m}^m]$ and its gradient stochastically via QMC in some cases~\cite{spencer_sign_2012,levy_mitigating_2019}, such approaches cannot yield certificates for the quality of the obtained basis as the average sign itself is not computed. 
Moreover, they are dependent on the distribution defined by $\abs{T_m}$ being well-behaved (i.e., ergodic and satisfying detailed balance) for QMC algorithms.

It therefore seems infeasible to find a converging and efficient algorithm for minimizing the average sign for general Hamiltonians directly.
Ideally, one could find a simple quantity measuring the non-stoquasticity of the Hamiltonian which can be connected to the inverse average sign in a meaningful way while at the same time admitting efficient evaluation. 

\subsection{Case studies}
\label{sec:examples}

We now show that this hope is in vain in its most general formulation. 
Specifically, we provide an example of a Hamiltonian which has large positive entries but is nevertheless sign-problem free (has unit average sign) for specific choices of $\beta$ and $m$, as well as an example of an Hamiltonian that is close to stoquastic but incurs an arbitrarily small average sign for certain choices of $\beta$ and $m$ in a specific QMC procedure. 

Here, as throughout this work, whenever we consider systems of multiple qubits, for $A \in \mb C^{2 \times 2}$ we define 
\begin{align}
	A_i = A \otimes \id_{\{i\}^c} , 
\end{align}
to be the operator that acts as $A$ on qubit $i$ and trivially on its complement $\{ i\}^c$.

\begin{example}[Highly non-stoquastic but sign-problem free Hamiltonians]
	Let us define a Hamiltonian term acting on two qubits with label $i,j$ as 
	\begin{align}
		h_{i,j} = - \frac 12 (X_i X_j - Y_iY_j )  +  X_i . 
	\end{align}
	Then this Hamiltonian term is non-stoquastic with total weight 
	$\nu_1(h_{i,j})  = 1 $. 
	What is more, the $n$-qubit Hamiltonian 
	\begin{align}
		H = \id + \sum_{i< j}^n h_{i,j}
	\end{align}
	is highly non-stoquastic with total weight $\nu_1(H) =n$.
	At the same time, the QMC algorithm for computing the partition function of $H$ with parameters $\beta, m$, has average sign $\l \sign \r_{\beta,m} = 1$. 
\end{example}

\begin{proof}
	We first determine the non-stoquasticity of $H$ as 
	\begin{align}
		\nu_1(H) = \sum_i \nu_1(X_i) = n . 
	\end{align}
	To see why the QMC algorithm has unit average sign, note that the transfer matrix $T_m = \id - \beta H /m$ has negative entries $T_m(\lambda|\lambda') < 0$ only if the parity of $\lambda \oplus \lambda'$ is odd since for these terms only a single $X$ term contributes. 
	Whenever $\lambda \oplus \lambda' = 0 $, i.e., has even parity, we have $T_m(\lambda | \lambda')\geq 0$ since only $XX -YY$ terms or the diagonal contribute -- both of which have non-negative matrix elements. 
	
	In the calculation of the partition function, the summation runs over closed paths only. 
	But for any closed path $\lambda_1 \rightarrow \lambda_2 \rightarrow \cdots \rightarrow \lambda_m \rightarrow \lambda_1$, it is necessary that the total parity $(\lambda_1 \oplus \lambda_2) \oplus \ldots \oplus (\lambda_m \oplus \lambda_1) $ vanishes. 
	In particular, this implies that every allowed path incurs an even number of odd-parity steps and therefore an even number of negative signs. 
	Therefore, only non-negative paths contribute to the path integral and the average sign is attained at unity. 
\end{proof}

\begin{example}[Barely non-stoquastic Hamiltonians with arbitrarily small average sign]
Let us define the $2$-qubit Hamiltonian 
\begin{align}
	H_{a,b} = & \frac m \beta \bigg( \id \otimes \id - \id \otimes X 
	- \frac 12( X \otimes X + Y\otimes Y) \\ 
	& + \frac 12 [ (a + b) X \otimes Z + (b - 
	a) X \otimes \id] \bigg),
\end{align}
with $b\geq a >0$ positive numbers and $m \in 2\mathbb{N} + 1$ .
The non-stoquasticity of $H_{a,b}$ is given by $\nu_1(H_{a,b}) = b m/(2 \beta)$, the average sign of QMC with parameters $\beta$ and $m$ is dominated by $|\l \sign \r_{\beta,m}|  \leq C (b - a)/{a}$, where $C$ is a constant.
Thus, even for arbitrarily small non-stoquasticity we can make the sign problem unboundedly severe as we tune $a$ to be close to $b$.
\end{example}

\begin{proof}
We derive the bound on the average sign.
For the given Hamiltonian, the corresponding transfer matrix for a QMC algorithm for inverse temperature $\beta$ with $m$ steps is given by
\begin{equation}
	T_m \equiv T_{a,b} = \begin{pmatrix}
		0 & 1 & -b & 0 \\
		1 & 0 & 1 & a \\
		-b & 1 & 0 & 1 \\
		0 & a & 1 & 0 \\
	\end{pmatrix}.
\end{equation}

Recall that the average sign can be written as 
\begin{equation}
	\l\sign\r_{\beta,m} = \frac{\tr[T_m^m]}{\tr[{|T_m|}^m]}. 
\end{equation}
We denote by $\overline T_m$ a matrix similar to $T_m$ but where the positions of $a$ and $-b$ are exchanged. Due to the symmetry of the problem we have that $\tr[T_m^m] = \tr[\overline T_m^m]$ and $\tr[|T_m|^m] = \tr[{|\overline T_m|}^m]$. 
Hence, 
\begin{align}
	\tr[T_m^m] &= \frac12\left[ \tr[T_m^m] + \tr[\overline{T}_m^m]\right] \\
		&= \frac12 \sum_{\vec\lambda \in \Lambda_m} \Bigg[ T_m(\lambda_1\mid\lambda_2) \cdots T_m(\lambda_m \mid \lambda_1)\\
		&\quad\quad\quad + \overline{T}_m(\lambda_1\mid\lambda_2) \cdots \overline{T}_m(\lambda_m \mid \lambda_1) \Bigg] \\
		&= \frac12  \sum_{\vec \lambda} \left[ a^{f(\vec\lambda)}(-b)^{g(\vec\lambda)} + a^{g(\lambda)}(-b)^{f(\vec\lambda)} \right],
\end{align}
where in the last line we have used the fact that every summand is a polynomial in the entries of $T_{a,b}$. 
The functions $g,f: \Lambda_m \to [m]$ describe the corresponding exponents. 
A little thought reveals that since all path are closed and $m$ is odd $g(\vec \lambda) + f(\vec \lambda)$ is larger than $1$ and also odd for all $\vec \lambda$. 
We thus find that one of the two terms for each $\vec\lambda$ must be negative and 
\begin{align}
	|\tr T_m^m| &\leq \frac12 \sum_{\vec\lambda} |a^{f(\vec\lambda)} b^{g(\vec\lambda)} - b^{f(\lambda)} a^{g(\vec\lambda)} | \\
	& \leq \frac12 \sum_{\vec\lambda}(2^{g(\vec\lambda)} - 1) a^{f(\vec\lambda)+g(\vec\lambda) - 1} | b - a |. 
\end{align}
Furthermore, we have 
\begin{align}
	\left|\tr |T_m|^m \right| &= \frac12 \sum_{\vec\lambda} (a^{f(\vec\lambda)} b^{g(\vec\lambda)} + b^{f(\lambda)} a^{g(\vec\lambda)} ) \\
	&\geq \sum_{\vec\lambda} \left(a^{f(\vec\lambda)+g(\vec\lambda) }\right).
\end{align}
Combining these two bounds and using $g(\vec \lambda) \leq m$, we conclude that 
\begin{equation}
	|\langle\sign \rangle | \leq \left(2^{m-1} - \frac12\right)\frac{|b-a|}a.
\end{equation}
\end{proof}

The second example shows that in principle also Hamiltonians with arbitrarily small positive entries can cause a severe increase of the sampling complexity of specific Monte Carlo algorithms. Interestingly, in this situation the sign problem cannot be eased by a basis change: the average sign vanishes since the \emph{unitarily invariant} term $|\tr T^m_m|$ is tuned to be small. 
On the contrary, the sign problem can be completely avoided in this example by choosing the Monte-Carlo path length to be even instead of odd. 

These simple examples illustrate the following important observation: The sign problem as measured by the average sign can in certain situations be avoided or intensified by fine-tuning the problem and parameters of the QMC procedure independently of the actual magnitude of the positive entries of the Hamiltonian. 
But such examples seem to rely on an intricate conspiracy of the structure of the Hamiltonian and the chooen QMC procedure, e.g., the discretization. It is plausible to assume that the most pathological cases are unlikely to appear in practical applications, and can at least be rather easily overcome by slightly modifying the QMC algorithm.

\subsection{Measures of non-stoquasticity}
\label{sec:motivating non-stoquasticity}

In this work, our goal is to develop a more general methodology for the task of easing the sign problem that is independent of the details of the QMC algorithm and the combinatorial properties of potential paths that can be constructed from the entries of the transfer matrix.
Very much in the spirit of the notion of stoquasticity we aim at a property of the Hamiltonian in a given basis to measure its deviation from having a good sampling complexity.
Natural candidates for such a non-stoquasticity measure of a Hamiltonian are entry-wise norms of its positive entries. 
For any $p \geq 1$ we define the non-stoquasticitiy of $H$ as 
\begin{align}
	\nu_p(H)  = D^{-1} \norm{\nonstoq H}_{\ell_p} , 
\end{align}
where $\norm{\cdot}_{\ell_p}$ denotes the vector-$\ell_p$ norm. 
For every $p$, $\nu_p$ is efficiently computable for local Hamiltonians on bounded-degree graphs and therefore suitable for easing the sign problem of a large class of Hamiltonians by local basis choices.  
It is also obviously a measure of the non-stoquasticity in the sense that $\nu_p(H) = 0$ if and only if $H$ is stoquastic. 
We note that we have chosen our definition such that the non-stoquasticity measure $\nu_p$ scales extensively in the number of non-stoquastic terms of a local Hamiltonian.
This is because every non-stoquastic local Hamiltonian term creates on the order of $2^n$ positive matrix entries in a global $n$-qubit Hamiltonian matrix due to tensoring with identities on the complement of its support. 

Our examples in the previous section show that it is notoriously difficult if not impossible to connect any notion of non-stoquasticity to the actual sample complexity incurred by a QMC procedure as measured by the inverse average sign. 
This is due to the dependence of the average sign on the combinatorics of Monte Carlo paths. 
However, those examples involved a large degree of fine-tuning. 
Therefore, one might hope to find a connection between non-stoquasticity and average sign for \emph{generic} cases.

So let us look at the connection between optimizing a non-stoquasticity measure $\nu_p$ and optimizing the QMC sampling complexity as in \eqref{eq:relative error minimization problem}. 
Our measure can be expressed in terms of the transfer matrix $T_m$ as 
\begin{equation}
\label{eq:nu_p}
	\nu_p(H) = \frac 1 D \frac  m {2\beta}  \norm{\abs{T_m} - T_m}_{\ell_p},
\end{equation}
where we assume that $\diag (\beta H/m) \leq \id$. 

Due to the unitary invariance of the trace, the optimization of the sampling complexity via \eqref{eq:relative error minimization problem} is equivalent to minimising 
\begin{equation}
	S(U) = \tr[|U T_m U^\dagger|^m] - \tr[T_m^m]. 
\end{equation}
Let us for the sake of clarity, suppress the explicit dependence on the unitary $U$ and define $\hat T_m = U T_m U^\dagger$.
If we define the positive and negative entries of the transfer matrix respectively as $\Delta_\pm = \frac12 \left(|\hat T_m| \pm \hat T_m \right)$, we can write 
\begin{align}
	S(U) &= \tr[|\hat T_m|^m] - \tr[\hat T_m^m] \\
	&= 2 \sum_{\substack{\vec s \in \{\pm\}^m: \\ \text{$\vec s$ odd}}} \tr[ \Delta_{s_1}\cdots \Delta_{s_m} ].
	\label{eq:s(u)}
\end{align}
The summation in the last line is restricted to all $\vec s \in \{\pm\}^m$ with an odd number of negative signs.
The resulting expression thus involves a summation over closed paths that contain an odd number of negative contributions such that $\Delta_{s_1}(\lambda_1|\lambda_2)\cdots \Delta_{s_m}(\lambda_m|\lambda_1) <  0 $. 
In particular, every such path contains at least one step with a negative contribution. 

The size of $S(U)$ thus depends both on the number of `negative paths' and their individual weight. 
From Eq.~\eqref{eq:s(u)} we find that the contribution to $S(U)$ of paths with exactly one negative step has the form
\begin{align}
    2 m \sum_{\lambda_1, \lambda_2} \Delta_{-}(\lambda_1 | \lambda_2) \Delta_{+}^{m-1} (\lambda_2 | \lambda_1) , 
	\label{eq:one-step path integral}
\end{align}
using the cyclicity of the trace. This expression \eqref{eq:one-step path integral} is a weighted sum over the negative entries of $\hat T_m$, where the weights are given by the contribution $\Delta_{+}^{m-1} (\lambda_2 | \lambda_1)$ of all positive paths of length $m-1$. 

For a transfer matrix in which the positive entries do not significantly differ and their distribution relative to the negative entries is unstructured, we have constant $\Delta_{+}^{m-1} (\lambda_2 | \lambda_1) \approx  \norm{\Delta_+^{m-1}}_{\ell_\infty}$. 
Therefore, the linear term \eqref{eq:one-step path integral} scales approximately as
\begin{align}
	2m \norm{\Delta_-}_{\ell_1}\norm{\Delta_+^{m-1}}_{\ell_\infty} \propto D \, \nu_1(H).  
\end{align}

For higher-order negative contributions, we expect that $S(U)$ or, correspondingly, the average sign scales as $\exp(c \cdot D \, \nu_1(H))$ for some $c > 0$. 
Our expectation is based on the following observation: 
in the calculation of the inverse average sign, all paths of length $m$ with an odd number of negative steps contribute. 
Potentially, in each step every negative entry of $T_m$ appears. Then the sum of all negative entries of $T_m$ contributes. 
But the number of paths with $k \in 2\mb N_0 +1$ negative steps scales as $\binom{m}{k}$ which leads to an exponential growth in $\norm{\nonstoq H}_{\ell_1} $ and hence $D \, \nu_1(H)$. 
In the following section, we provide a brief numerical analysis confirming this expectation.

We further observe that if the positive entries of $\hat T_m$ are more structured, the weights appearing in Eq.~\eqref{eq:one-step path integral} might deviate from a uniform distribution. 
In such a case, other $\nu_p$-measures become meaningful as a measure of the inverse average sign since they saturate a corresponding H{\"o}lder bound. 

\subsection{Numerical analysis}
\label{sec:av sign numerics}

\begin{figure}[t]
	\includegraphics[width = \linewidth]{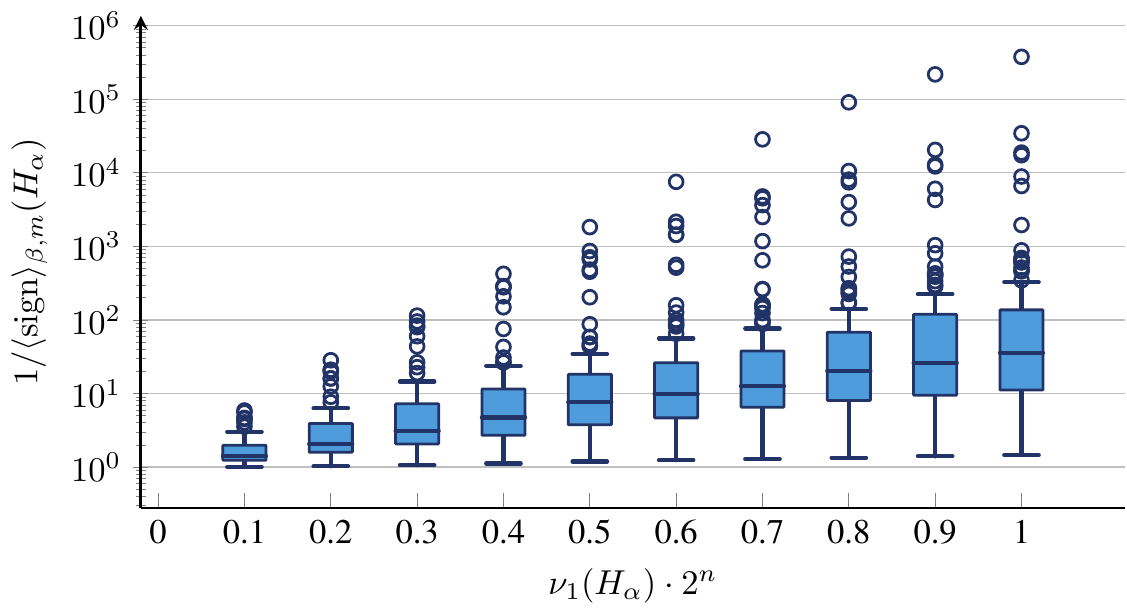}
	\caption{The figure shows the inverse average sign for $100$ randomly chosen instances of $5$-qubit Hamiltonians $H_\alpha$ for $\beta = 1$ and $ m = 100$ Monte Carlo steps as a function of $d \nu_1(H_\alpha) $. 
	We find a roughly exponential dependence of the inverse average sign with $\nu_1(H_\alpha)$ as $1/\l \sign \r_{\beta,m}(H_\alpha) \propto \exp(a \cdot d \nu_1(H_{\alpha}))$ for $a > 0 $. 
	\label{fig:avsign numerics}
	 }
\end{figure}

In this subsection, we provide evidence that $\nu_1(H)$ is indeed a very much meaningful measure of the sample complexity and hence the inverse average sign by exactly calculating the inverse average sign as a function of $\nu_1(H)$.
We do so by randomly drawing $2$-local Hamiltonians on a line of $n$ qubits of the form 
\begin{align}
	\label{eq:ti h}
	H = \sum_{i = 1}^{n}T_i(h) , 
\end{align}
where $h \in \mb R^{4 \times 4}$ is a nearest-neighbour interaction term and the translation operator $T_i $ acts as $T_i(h) = \id_d^{\otimes (i-1)} \otimes h \otimes  \id_d^{\otimes n - i - 1 }$. 
We choose each local term $h$ in an i.i.d.\  fashion
from the zero-centered Gaussian measure and projecting to Hermitian matrices. 
For each random instance $H$, 
we then consider the one-parameter Hamiltonian family 
\begin{align}
	H_\alpha = \frac{  H - \nonstoq H  + \alpha \nonstoq H }{2^n \nu_1( \nonstoq H)}. 
\end{align}
Note that $\nu_1(H_\alpha) = \alpha/2^n$. 
Fig.~\ref{fig:avsign numerics} shows that, generically, the average sign monotonously depends on the non-stoquasticity. 
Indeed, as expected for large $m$, the dependence is an exponential one.

\subsection{Computing the non-stoquasticity}
\label{sec:computing non-stoq}

Above, we have argued that a key problem of the average sign lies in the fact that it is not efficiently computable whenever there is a sign problem. 
But how does the non-stoquasticity measure $\nu_1$ fare in this respect? 
We now show that for arbitrary $2$-local Hamiltonians the non-stoquasticity measure $\nu_1$ can in fact be approximated up to an inverse polynomially small additive error in polynomial time. 
While this is not possible for arbitrary local Hamiltonians as a simple example shows, any $\nu_p$-measure can be efficiently computed exactly in polynomial time for local Hamiltonians acting on bounded-degree graphs. 


We write a real $2$-local Hamiltonian with $1$-local terms as
\begin{align}
\label{eq:two-local h}
	H_{2+1} = & \sum_{i < j }  
	\bigl(a_{i,j} X_i X_j +  b_{i,j} Y_i Y_j + c_{i,j} Z_i Z_j \\&  + x_{i,j} X_i Z_j + x_{j,i} Z_i X_j\bigr) + \sum_i \bigl( \alpha_i X_i + \gamma_i Z_i
	\bigr)\, ,
	\nonumber 
\end{align} 
parametrized by real coefficient vectors $a,b,c \in \mb R^{n(n-1)/2}, \, x \in \mb R^{n(n-1)}$ which are non-zero only on the edges $(i,j) \in E$ of the interaction Hamiltonian graph $G = (V,E)$ as well as vectors $\alpha, \gamma \in \mb R^n$. 
The Hamiltonian interaction graph is defined by a set $V$ of sites or vertices and the edge set 
\begin{multline}		
	 E= \big\{ (i,j) \in V\times V:\\ \neg ( a_{i,j}=  b_{i,j} =  c_{i,j} =  x_{i,j} =  x_{j,i} = 0 )\big \}. 
\end{multline}
We call $\mc N(i) = \{ j : (i,j) \in E\} $ the neighbourhood of site $i$ on the graph $G$, $\deg(i) = |\mc N(i)|$ the degree of site $i$ and $\deg(G) = \max_{i \in V} \deg(i)$ the degree of the overall graph. 
Notice that obtaining an expression for the non-stoquasticity is non-trivial since several local Hamiltonian terms may contribute to a particular entry of the global Hamiltonian matrix. 

\begin{lemma}[Non-stoquasticity of $(2+1)$-local Hamiltonians]
\label{lem:non-stoquasticity 2-local H}
	The non-stoquasticity measure $\nu_1$ of real $2$-local Hamiltonians with $1$-local terms of the form $H_{2+1}$ satisfies 
	\begin{align}
	\begin{split}
	\nu_1 (H_{2+1}) = & \sum_{i<j} \nu_1(a_{i,j} X_i X_j + b_{i,j} Y_i Y_j)\\& + \sum_i \nu_1\bigg(\alpha_i X_i + \sum_{j\in \mc N_{XZ}(i)} x_{i,j} X_i Z_j \bigg) , 
	\end{split}
	\end{align}
	and it holds that 
\begin{align}
	& \nu_1(a_{i,j} X_iX_j + b_{i,j}Y_i Y_j) =	 \label{eq:non-stoquasticity xx+yy} \\
	& \qquad  \frac 12 \sum_{i < j } \bigl( \max\{a_{i,j}+ b_{i,j},0\} + \max\{a_{i,j}- b_{i,j},0\} \bigr), 
\nonumber \\
 \begin{split}
 	\label{eq:non-stoquasticity xz+x}
	& \nu_1\bigg( \alpha_i X_i + \sum_{j \in \mc N_{XZ}(i)} x_{i,j} X_i Z_j \bigg) =  2^{- \deg_{XZ}(i)} \\ 
	&\qquad\times  \sum_{\lambda_{\mc N_{XZ}(i)}} \max  \bigg\{\alpha_i + \sum_{j \in \mc N_{XZ}(i)}(-1)^{\lambda_j} x_{i,j} ,0\bigg\} . 
	 \end{split}
\end{align}	
\end{lemma}		

Here, we have defined the $XZ$ neighbourhood $\mc N_{XZ}(i) = \{ j: x_{ij} \neq 0 \} $ of site $i$ as all vertices $j$ connected to $i$ by an $X_i Z_j$-edge. 
As useful shorthands, we also define the $XZ$ degree $\deg_{XZ}(i) = |\mc N_{XZ}(i)|$ and the restriction of a spin configuration $\lambda \in \{0,1\}^n$ to an $XZ$ neighbourhood as  $\lambda_{\mc N_{XZ}(i)} = (\lambda_j)_{j \in \mc N_{XZ}(i)} \in \{0,1\}^{\deg_{XZ}(i)}$. 
We conceive of summation over an empty set (the case that $\mc N_{XZ}(i) = \{\}$) as resulting in $0$ so that the corresponding term in Eq.~\eqref{eq:non-stoquasticity xz+x} vanishes. 

Notice that the non-stoquasticity of an $XZ$ term does not depend on the sign of its weight, while for $XX$ and $X$ terms it crucially does. 
 
\begin{proof}
We can determine the $\ell_1$-norm of the off-diagonal part of the Hamiltonian $H_{2+1}$ as
\begin{equation}
\begin{split}
	\norm{H&_{2+1}  -  \diag(H_{2+1})}_{\ell_1}  \\
	= & \sum_{\lambda \in \{0,1\}^n} \bigg\{ \sum_{i < j } |a_{i,j} -  (-1)^{\lambda_{i} + \lambda_j}b_{i,j}| \\
	& + \sum_i \bigg |\alpha_i + \sum_{j \in \mc N_{XZ}(i)} (-1)^{\lambda_j} x_{i,j} \bigg|  \bigg\}\\
	=  & 2^{n-1}  \sum_{i < j } \bigl( |a_{i,j} + b_{i,j}| + |a_{i,j} - b_{i,j}| \bigr) \\
	 & \hspace{-1ex} +   \sum_i  2^{n- \deg_{XZ}(i)}  \sum_{\lambda_{\mc N_{XZ}(i)} } \bigg |\alpha_i + \sum_{j \in \mc N_{XZ}(i)} (-1)^{\lambda_j} x_{i,j} \bigg| . \label{eq:ell_1 norm H2}
	\end{split}
\end{equation}
From Eq.~\eqref{eq:ell_1 norm H2} we can then directly calculate the non-stoquasticity $\nu_1$ of $H_2$ 
by discarding all matrix entries with negative sign before taking the $\ell_1$-norm and dividing by $2^n$. 
\end{proof}

Now, clearly we can efficiently compute the term~\eqref{eq:non-stoquasticity xx+yy} for arbitrary graphs as the sum runs over at most $n(n-1)/2$ many terms. 
In the term~\eqref{eq:non-stoquasticity xz+x}, in contrast, the sum over spin configurations $\lambda_{\mc N_{XZ}(i)}$ in the $XZ$ neighbourhood of site $i$ runs over $2^{\deg_{XZ}(i)}$ many terms and hence this term is efficiently computable exactly in case the vertex degree $\deg_{XZ}(i)$ of any vertex $i$ grows at most logarithmically with $n$. 
This shows that for bounded-degree graphs such as regular lattices the non-stoquasticity can be computed efficiently. 

But what if the degree of the graph grows faster than logarithmically with $n$ so that the sum runs over super-polynomially many non-trivial terms?
The following Lemma shows that even in this case, that is, for $2$-local Hamiltonians acting on arbitrary graphs, the non-stoquasticity can be efficiently approximated up to any inverse polynomially small additive error using Monte Carlo sampling. 

\begin{theorem}
\label{thm:approximating non-stoq arb graph}
	The sum~\eqref{eq:non-stoquasticity xz+x} can be efficiently approximated up to additive error $\epsilon$ via Monte Carlo sampling with failure probability $\delta $ from $16 \deg_{XZ}(i) (\max_{j} |x_{i,j}|)^2 \log(2/\delta) /\epsilon^2 $ many iid.\ samples. 
\end{theorem}
\begin{proof}
	For the proof we will use concentration of measure for Lipschitz functions.
	To this end we begin by noticing that the sum~\eqref{eq:non-stoquasticity xz+x} can be rewritten as a uniform expectation value over $k_i = \deg_{XZ}(i)$ many Rademacher ($\pm 1$) random variables as 
\begin{align}
\begin{split}
	\sum_{\lambda_{\mc N_{XZ}(i)}} & \max  \bigg\{\alpha_i + \sum_{j \in \mc N_{XZ}(i)}(-1)^{\lambda_j} x_{i,j} ,0\bigg\} \\
	&  = \mb E_{\sigma \in \{\pm 1\}^{k_i}} [f^{(i)}_{\alpha, x}(\sigma)] , 
	\label{eq:rademacher expectation}
\end{split}
\end{align}
where we have defined 
\begin{align}
\begin{split}
	f^{(i)}_{\alpha, x}:\, & \mb R^{k_i} \rightarrow \mb R \\
	& s \mapsto \max \bigg \{\alpha_i + \sum_{j \in \mc N_{XZ}(i)} s_j x_{i,j} ,0 \bigg \}. 
\end{split}
\end{align}
It can easily be seen that the function $f^{(i)}_{\alpha,x}$ is Lipschitz with constant $\big( \max_j | x_{i,j} |\big) k_i^{ 1/2}$:  
\begin{align}
	\big | f^{(i)}_{\alpha,x}(s) & - f^{(i)}_{\alpha,x}(s') \big| = \bigg| \sum_{j = 1}^{\deg_{XZ}(i)} x_{i,j}\bigl( s_i - s_i' \bigr) \bigg |\\
	& \leq \big( \max_j | x_{i,j} |\big) \norm{s - s'}_{\ell_1}\\
	& \leq \big( \max_j | x_{i,j} |\big) k_i^{1/2} \norm{s - s'}_{\ell_2} .
\end{align}
Here, we have used the fact that the $\ell_p$ norms on $\mb R^n$ satisfy 
Moreover, $f^{(i)}_{\alpha,x}$ is clearly \emph{separately convex}, that is, for each $k = 1, 2, \ldots, k_i$ the function $s_j \mapsto f^{(i)}_{\alpha,x}(s_1, s_2, \ldots, s_{j-1}, s_j, s_{j+1}, s_{j+2},\ldots, s_n)$ is convex for each fixed $(s_1, s_2, \ldots, s_{j-1}, s_{j+1}, s_{j+2},\ldots, s_n) \in \mb R^{k_i - 1}$. 

We can then apply \cite[Theorem 3.4]{wainwright_high-dimensional_2019} to obtain that the estimator 
\begin{align}
\hat f_{\alpha, x}^{(i)} = \frac  1m \sum_{l =1}^m f^{(i)}_{\alpha, x}(\sigma^{(l)}), 
\end{align}
for the $m$ Rademacher vectors $\sigma^{(l)} \in \{ \pm 1\}^{k_i}$ drawn iid.\ uniformly at random satisfies
\begin{align}
	\mb P_\sigma \left [ \bigl | \hat f_{\alpha, x}^{(i)} - \mb E_\sigma[ f_{\alpha, x}^{(i)}\sigma ]\bigr | \geq \epsilon \right] \leq 2 \ee^{- \frac{ m\epsilon^2}{16 k_i (\max_j |x_{i,j}|)^2 }}. 
\end{align}
This implies that with probability $1 - \delta $ the estimator satisfies 
\begin{align}
	 \bigl | \hat f_{\alpha, x}^{(i)} - \mb E_\sigma[ f_{\alpha, x}^{(i)}\sigma ]\bigr | \leq \epsilon 
\end{align}
whenever the number $m$ of independently drawn Rademacher vectors satisfies
\begin{align}
	m \geq  \frac{16 \,k_i (\max_j |x_{i,j}|)^2 }{\epsilon^2} \log(2/\delta). 
\end{align}
\end{proof}

In total we thus obtain a polynomial worst-case complexity of computing the non-stoquasticity of a $(2+1)$-local Hamiltonian of the form~\eqref{eq:two-local h} up to additive error $\epsilon$ with failure probability $\delta$ as given by
\begin{align}
	\frac{n (n-1)}{2} + \frac{16 \sum_i\deg_{XZ}(i)(\max_{i,j}x_{i,j})^2}{(\epsilon/n)^2 } \log \left(\frac 2 \delta \right). 
	 \label{eq:complexity non-stoq 2-local}
\end{align}


\section{Easing the sign problem: An algorithmic approach}
\label{sec:practical easing}

To demonstrate the feasibility of \se\, and to put our findings more closely into the context of
practical condensed matter problems, 
we numerically optimize the non-stoquasticity of certain nearest-neighbour Hamiltonians in quasi one-dimensional ladder geometries. 
Such systems are effectively described by translation-invariant 
Hamiltonians on $n$ $d$-dimensional quantum systems of the form \eqref{eq:ti h} with nearest-neighbour interaction term $h \in \mb R^{d^2 \times d^2}$. 
For the sake of simplicity, we specialize here to closed boundary conditions, identifying $n +1 = 1 $. 

We then optimize the non-stoquasticity of $H$ over on-site orthogonal basis choices of the type
\begin{align}
\label{eq:local basis choice}
	H  \mapsto O^{\otimes n} H (O^T)^{\otimes n} . 
\end{align}
On-site transformations are particularly simple to handle as they preserve both locality and translation-invariance of the Hamiltonian. 
Due to the translation-invariance of the problem 
the global non-stoquasticity measure can be expressed locally in terms of the transformed term $O^{\otimes 2 } h (O^T)^{\otimes 2} $ 
so that the problem has constant complexity in the system size. 
More precisely, for Hamiltonians of the form \eqref{eq:ti h} we can express the non-stoquasticity measure $\nu_1(H)  = n 2^{n - 3} \tilde \nu_1(h)$ in terms of an effective local measure 
\begin{align}
	\tilde \nu_1(h) =  \sum_{\substack{ijk;lmn: \\j\neq m, k = n}} \max \left\{ 0,\left(h \otimes \id + \id \otimes h  \right)_{ijk; lmn}\right\} . 
\end{align}
Optimizing $\nu_1(H)$ for the global Hamiltonian is therefore equivalent to the much smaller problem of  minimizing $\tilde \nu_1(h)$. 
While the non-stoquasticity measure $\nu_1$ can be efficiently evaluated, thus satisfying a necessary criterion for an efficient optimization algorithm, minimizing $\nu_1$ may and in fact \emph{will} still be a non-trivial task in general -- an intuition we make rigorous below.
This is because in optimizing the basis-dependent measure $\nu_1$ over quasi-local basis choices one is faced with a highly non-convex optimization problem of a high-order polynomial over the sphere of orthogonal matrices.  
Among the best developed multi-purpose methods for optimization over group manifolds such as the orthogonal group are conjugate gradient descent methods \cite{abrudan_conjugate_2009}. 
Compared to simple gradient-descent algorithms, conjugate gradient algorithms are capable to better incorporate the underlying group structure to the effect that they achieve much faster runtimes and better convergence properties. 

\begin{figure}
\centering
	\includegraphics{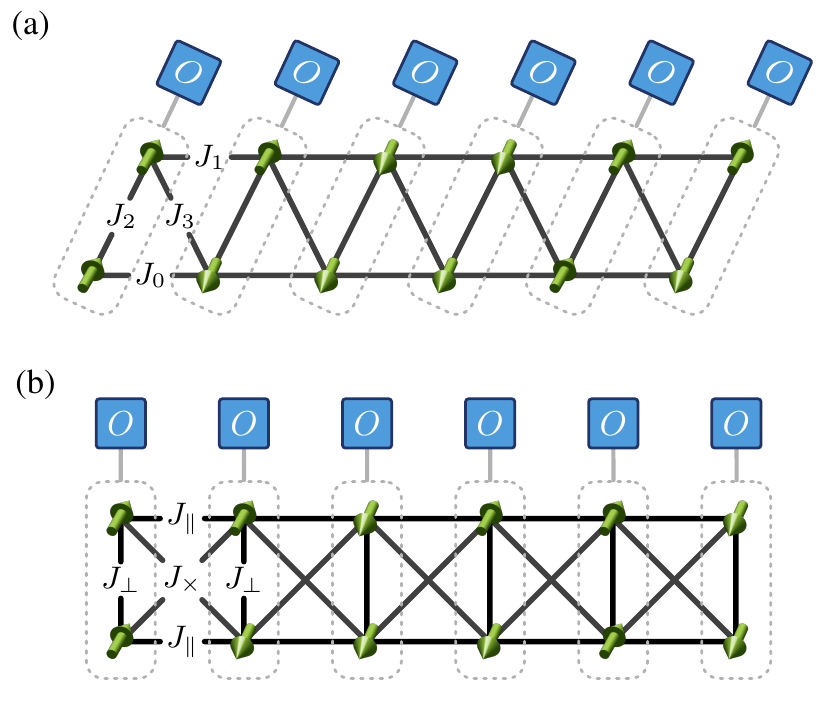}
	\caption{
	Quasi one-dimensional models with a sign problem. 
	Figure \emph{(a)} shows the lattice of the \jmodel-Heisenberg model on a triangular quasi one-dimensional lattice as given in Eq.~\eqref{eq:jmodel hamiltonian}. 
	Figure \emph{(b)} shows the lattice structure of the frustrated Heisenberg model \eqref{eq:frustrated hamiltonian} with couplings $J_\perp$, $J_\parallel$ and $J_\times$ on a square-lattice ladder with cross coupling. 
	In our simulations, we group sites to dimers as indicated in the figures and then optimize the measure $\tilde \nu_1(h)$ of the effective $2$-local terms $h$ over on-site orthogonal transformations $O^{\otimes 2}$. 
	\label{fig:ladders}
	}
\end{figure}

To practically minimize the non-stoquasticity $\tilde \nu_1$ over the orthogonal group $O(d)$ we have implemented a conjugate gradient descent algorithm following Ref.~\cite{abrudan_conjugate_2009}. 
Our implementation is publicly available~\cite{optimization_numerics} and detailed in Appendix~\ref{app:conjugate gradient}. 

We first benchmark the algorithm on Hamiltonians which we know to admit an on-site stoquastic basis. 
Specifically, we apply the algorithm to recover an on-site 
stoquastic basis of the random translation-invariant Hamiltonian
\begin{align}
\label{eq:random stoquastic h}
H = \sum_{i = 1}^{n} T_i\left ( O^{\otimes 2} (h - \nonstoq h) (O^T)^{\otimes 2} \right)	
\end{align} 
on $n$ qudits where the two-local term $h \in \mb R^{d^2 \times d^2 }$ is a Hamiltonian term with uniformly random spectrum expressed in a Haar-random basis 
and $O \in O(d)$ is a Haar-random on-site orthogonal matrix. 
In Fig.~\ref{fig:stoquastic plots}(a) we show the performance of the algorithm on randomly chosen instances of \eqref{eq:random stoquastic h} for different values of the local dimension $d$. 
In all but very few instances our algorithm essentially recovers the stoquastic basis of the random Hamiltonian. 
By construction, this can only be due to the fact that the algorithm gets stuck in local minima, indicating a potential limitation of first-order optimization techniques as a tool for easing the sign problem of general Hamiltonians. 

We then study frustrated anti-ferromagnetic Heisenberg Hamiltonians, i.e., Hamiltonians with positively weighted interaction terms $\vec S_i \cdot \vec S_j$, on different ladder geometries. 
Here, $\vec S_i = (X_i,Y_i,Z_i)^T$ is the spin operator at site $i$. 
The sign problem of frustrated ladder systems can in many cases actually be removed by going to a dimer basis \cite{nakamura_vanishing_1998,honecker_thermodynamic_2016,wessel_efficient_2017}. 
However -- and this is important for our approach -- in those cases the sign problem is not removed by finding a \emph{stoquastic local basis}, but rather by exploiting specific properties of the Monte Carlo simulation at hand, for example, that no negative paths occur in the simulation \cite{nakamura_vanishing_1998} or by exploiting specific properties of the Monte Carlo implementation at hand \cite{honecker_thermodynamic_2016,wessel_efficient_2017}. 
Therefore, frustrated Heisenberg ladders constitute the ideal playground to explore the methodology of easing the sign problem by (quasi-)local basis choices. 

\begin{figure*}[t]
	\centering
	\includegraphics{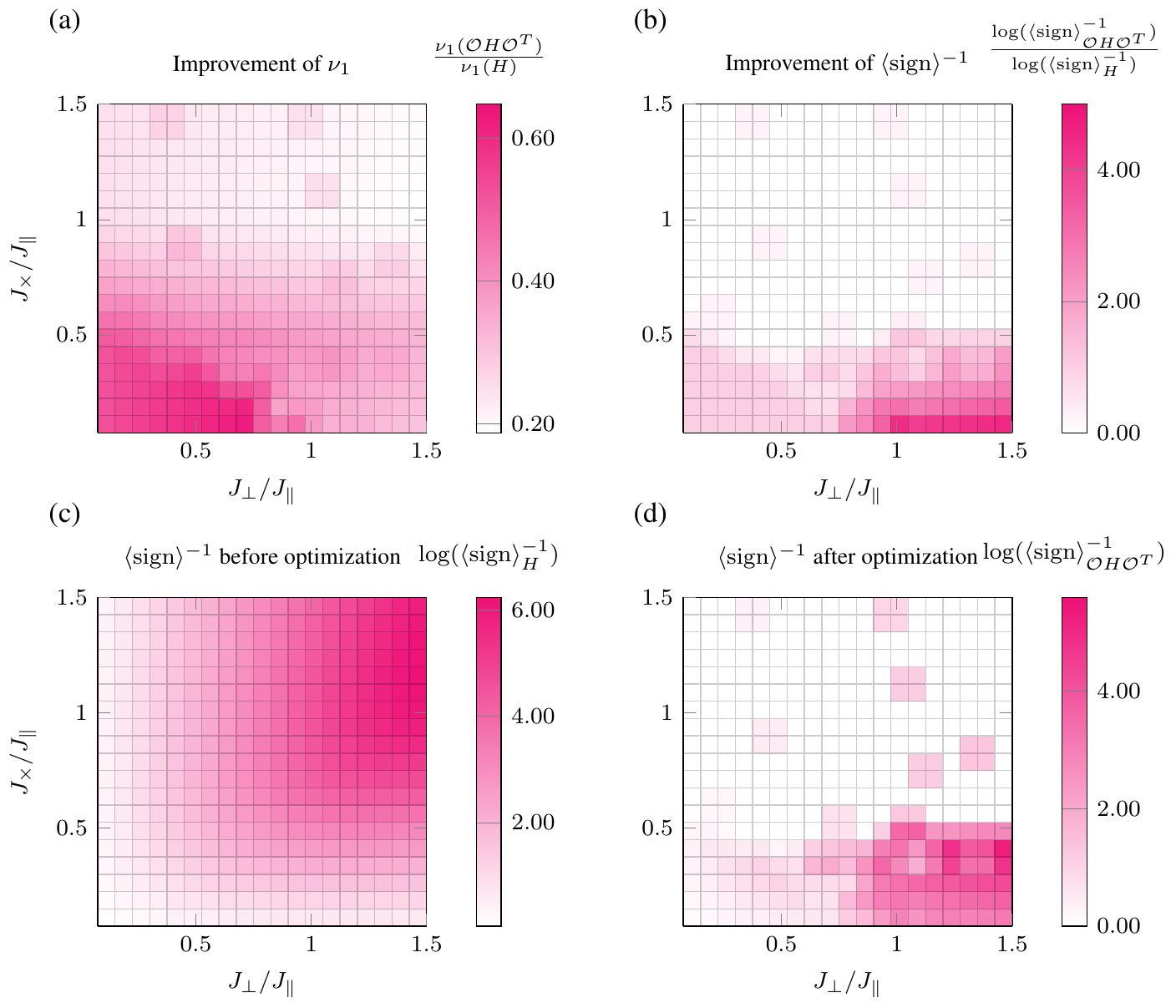}
\caption{
\label{fig:frustrated ladder av sign} 
Improvement of the average sign concomitant with an improvement of the non-stoquasticity of the frustrated Heisenberg ladder~\eqref{eq:frustrated hamiltonian}. 
Figure 
  \emph{(a)} shows the optimized non-stoquasticity $\nu_1$ in terms of its relative improvement compared to the computational basis. 
  \emph{(b)} We expect the inverse average sign to scale exponentially in the non-stoquasticity. Therefore, we plot the ratio of the logarithm of the inverse average sign before optimization to that after optimization. 
  We compute the average sign via exact diagonalization for a ladder of $2 \times 4$-sites, $m = 100$ Monte Carlo steps and inverse temperature $\beta = 1$.
  We also plot the logarithm of the inverse average sign \emph{(c)} before and \emph{(d)} after optimization of a local orthogonal basis. 
 }
\end{figure*}

The first model we study is the \jmodel-model \cite{nakamura_vanishing_1998} whose Hamiltonian is given by (see Fig.~\ref{fig:ladders}(a))
\begin{align}
\label{eq:jmodel hamiltonian}
	H_{\vec J} = \sum_{i = 1}^n \bigg(J_0 \vec S^1_i \vec S^1_{i+1} +  J_1 \vec S^2_i \vec S^2_{i+1} +  J_2 \vec S^1_i \vec S^2_i +  J_3\vec S^1_{i+1} \vec S^2_i \bigg), 
\end{align}
where $\vec S^1_i$ denotes the spin operator at site $i$ on the lower rung and $\vec S^2_i$ on the upper rung of the ladder, respectively, and $J_i \geq 0 $ for all $i$. 
Intriguingly, this Hamiltonian does not have a sign problem in the singlet-triplet dimer basis even though the Hamiltonian is not stoquastic in that basis. 
However, there exists a \emph{non-local} stoquastic basis for values of $J_3 \geq J_0 + J_1$~\cite{nakamura_vanishing_1998}.
We show the results of optimizing the non-stoquasticity of $H_{\vec J}$ with $J_0 = J_1 = J$ over a translation-invariant dimer basis (see Fig.~\ref{fig:ladders}(a)) in Fig.~\ref{fig:stoquastic plots}(b). 
We initialize our simulations in a Haar random orthogonal on-site basis. 
Interestingly, we find an improvement of the non-stoquasticity under on-site orthogonal basis choices that does not seem to correlate with the region in which a non-local stoquastic basis was found in Ref.~\cite{nakamura_vanishing_1998}. 
We view this as an indication that less local ansatz classes such as quasi-local circuits can further improve the non-stoquasticity for this model. 

We now apply the algorithm to the anti-ferromagnetic Heisenberg ladder studied in Refs.~\cite{honecker_thermodynamic_2016,wessel_efficient_2017}. 
The Hamiltonian of this system is given by (see Fig.~\ref{fig:ladders}(b))
\begin{align}
\label{eq:frustrated hamiltonian}
\begin{split}
	H_{J_\parallel,J_\perp,J_\times} = & \sum_{i = 1}^n \bigg(J_\parallel \left( \vec S^1_i \vec S^1_{i+1} +  \vec S^2_i \vec S^2_{i+1} \right)
	  +  J_\perp  \vec S^1_i \vec S^2_i \\
	  & \qquad + J_\times \left(\vec S^1_{i} \vec S^2_{i + 1}  + \vec S^1_{i + 1} \vec S^2_{i}  \right) \bigg), 
\end{split}
\end{align}
with interaction constants $J_\parallel,J_\perp,J_\times \geq 0 $. 
For this geometry, the situation is somewhat more involved: Refs.~\cite{honecker_thermodynamic_2016,wessel_efficient_2017} find that a sign-problem free QMC procedure exists, albeit for a slightly different QMC procedure than we consider here, namely stochastic series expansion (SSE) Monte Carlo \cite{sandvik_computational_2010}.
Similar to the world-line Monte Carlo method discussed here, SSE is based on an expansion of the exponential $\exp(- \beta H)$ albeit via a Taylor expansion as opposed to a product expansion. 
Specifically, for the partially frustrated case in which $J_\times \neq J_\parallel$ their solution of the sign problem exploits a specific sublattice structure of the Hamiltonian in combination with the SSE approach. 
We optimize the non-stoquasticity of dimer basis choices as shown in Fig.~\ref{fig:ladders}(b) when starting from a random initial point that is close to the identity.
Our results, shown in Fig.~\ref{fig:stoquastic plots}(c), qualitatively reflect the findings of \citet{wessel_efficient_2017} for SSE in terms of stoquasticity in that the non-stoquasticity can be significantly reduced for the fully frustrated case $J_\parallel = J_\times$, while it can be merely slightly improved for the partially frustrated case. 

At the same time, the algorithm does not recover a fully stoquastic basis for the frustrated ladder model $H_{J_\parallel,J_\perp,J_\times}$ as might be expected.
There may be several reasons for this: 
either the nearly sign-problem free QMC procedure found in Refs.~\cite{honecker_thermodynamic_2016,wessel_efficient_2017} is in fact specific to SSE in that no stoquastic dimer basis and hence no sign-problem free world-line Monte Carlo method exists in the orbit of orthogonal dimer bases, or the conjugate gradient algorithm gets stuck in local minima.  
In any case, the performance of our algorithm for both frustrated ladders demonstrates that the optimization landscape is generically an extremely rugged one, reflecting the computational hardness of the optimization problem in general. 

We now turn to showing 
the improvement of the average sign concomitant with the improvement in non-stoquasticity in Fig.~\ref{fig:frustrated ladder av sign}. 
We first observe that Figs.~\ref{fig:frustrated ladder av sign}(a) and (b) are compatible with an exponential dependence of the inverse average sign on the non-stoquasticity $\l \sign \r^{-1} \propto \exp(c \nu_1(H))$ as conjectured above: 
in the regions in which a significant improvement of the non-stoquasticity could be achieved by local basis choices, the inverse average sign could also be strongly improved. 
Importantly, while the Hamiltonian could not be made entirely stoquastic, the improvement in the inverse average sign reaches an extent to which nearly no sign problem remains in those regions. 
This shows that also moderate improvements in non-stoquasticity can yield tremendous improvements of the average sign. 
At the same time, a severe sign problem remains -- and actually becomes worse -- in a small region of the parameter space (around $J_\perp/J_\parallel \gtrsim 3/4$ and $J_\times/J_\parallel \lesssim 1/2$) even though the non-stoquasticity could be reduced to some extent in that region. 
This may reflect open questions about the relation between average sign and non-stoquasticity that arose in our earlier discussion in Section~\ref{sec:av sign vs nonstoq}: 
while in generic cases the two notions of severeness of the sign problem are expected to be closely related, there is no general simple correspondence between them. 

Our findings demonstrate both the feasibility of minimizing the non-stoquasticity in order to ease the sign problem by optimizing over suitably chosen ansatz classes of unitary/orthogonal transformations and potential obstacles to a universal solution of the sign problem. 
In particular, for translation-invariant problems -- while it may well be computationally infeasible -- the complexity of the optimization problem only scales with the locality of the Hamiltonian, the local dimension and the depth of the circuit. 
We expect, however, that there exists no algorithm with polynomial runtime in all of these parameters that always solves the optimisation problem. 

Our findings also indicate that more general ansatz classes yield the potential to further improve non-stoquasticity. 
Different classes of orthogonal transformations can be straightforwardly incorporated in our algorithmic approach. 
A detailed study of different ansatz classes and their potential for easing the sign problem is, however, beyond the scope of this work. 
It is the subject of ongoing and future efforts to study the optimal basis choice in terms of the non-stoquasticity for both deeper circuits and further models as well as the connection between the average sign and the non-stoquasticity. 


\section{Easing the sign problem: computational complexity }
\label{sec:complexity}

Let us now focus on a more fundamental question, namely, how far an approach that optimizes the non-stoquasticity can carry in principle. 
We have explored the potential of easing using state-of-the-art optimization algorithms; let us now turn to its boundaries, a glimpse of which we have already witnessed in the shape of a rugged optimization landscape. 
The method of choice for this task is the theory of computational complexity.

We analyze the computational complexity of easing the sign problem under particularly simple basis choices, namely, real on-site Clifford and orthogonal transformations. 
In both cases, we show that easing the sign problem is an \np-complete task even in cases in which deciding whether the sign problem can be cured is efficiently solvable \cite{klassen_two-local_2018}, namely for XYZ Hamiltonians as given by 
\begin{align}
\label{eq:xyz hamiltonian}
	H_{\text{XYZ}} = \sum_{i < j } \left(a_{i,j} X_i X_j + b_{i,j} Y_i Y_j + c_{i,j} Z_i Z_j \right).
\end{align}
Like Refs.~\cite{troyer_computational_2005,marvian_computational_2018,klassen_two-local_2018,klassen_hardness_2019}, we allow for arbitrary interaction graphs.

A central ingredient in proving Theorem~\ref{thm:complexity of se} is an expression for the non-stoquasticity measure $\nu_1$ of strictly $2$-local Hamiltonians of the form 
\begin{equation}
\begin{split}
\label{eq:strictly two-local h}
	H_2 = \sum_{i < j } & 
	\bigl(a_{i,j} X_i X_j +  c_{i,j} Z_i Z_j 
	 + x_{i,j} X_i Z_j + x_{j,i} Z_i X_j
	\bigr)\, . 
\end{split}
\end{equation} 
It is sufficient to restrict to Hamiltonians of the form~\eqref{eq:strictly two-local h} because the orbit of XYZ Hamiltonians under on-site orthogonal (Clifford) transformations does not reach $YY$ terms. 

It is a direct consequence of Lemma~\ref{lem:non-stoquasticity 2-local H} that 
\begin{align}
\begin{split}
\nu_1 & (H_2) = \sum_{i<j} \nu_1(a_{i,j} X_i X_j) + \sum_i \nu_1\bigg(\sum_{j\in \mc N_{XZ}(i)} x_{i,j} X_i Z_j \bigg) , 
\end{split}
\end{align}
where the $XZ$ neighbourhood $\mc N_{XZ}(i)$ of a vertex $i$ and related notions are defined in Sec.~\ref{sec:computing non-stoq}. 
More specifically, following Eqs.~\eqref{eq:non-stoquasticity xx+yy} and \eqref{eq:non-stoquasticity xz+x} we find that 
\begin{align}
 & \nu_1(a_{i,j} X_iX_j) =  \sum_{i < j } \max\{a_{i,j},0\}, 
 \label{eq:non-stoquasticity xx} \\
\begin{split}
	\label{eq:non-stoquasticity xz}
& \nu_1\bigg( \sum_{j \in \mc N_{XZ}(i)} x_{i,j} X_i Z_j \bigg) =  2^{- \deg_{XZ}(i)} \\ 
&\qquad\times  \sum_{\lambda_{\mc N_{XZ}(i)}} \max  \bigg\{\sum_{j \in \mc N_{XZ}(i)}(-1)^{\lambda_j} x_{i,j} ,0\bigg\} . 
 \end{split}
\end{align}	

Since for the proof of hardness we need analytical expressions of the non-stoquasticity, we cannot resort to the sampling algorithm to evaluate the non-stoquasticity of $XZ$ terms as proposed in Sec.~\ref{sec:computing non-stoq}. 
We analytically bound the contribution of a vertex with non-trivial $XZ$ neighbourhood with the following lemma. 
\begin{lemma}[$XZ$ non-stoquasticity]
\label{lem:xz-non-stoquasticity}
	The following bound is true for $k \in \mb N$
	\begin{align}
		\sum_{\lambda \in \{0,1\}^k} \max\bigg\{ \sum_{j= 1}^k (-1)^{\lambda_j} x_j ,0 \bigg\} \geq \max_j |x_j| \cdot 2^{k-1} .  
		\label{eq:nonstoq arbitrary xz vertex}
	\end{align} 
\end{lemma}
\begin{proof}
	Let us assume wlog.\ that $x_1 \ge x_2 \ge \ldots \ge x_k \ge 0 $, all terms being positive and non-increasingly ordered. 
	This does not restrict generality as all possible combinations of signs appear in the sum \eqref{eq:nonstoq arbitrary xz vertex}. 
	We prove the claim by induction. 
	For $k=1$, the statement is true by immediate inspection. 
	For the induction step, we use the following inequality for $a,b \in \mb R$
	\begin{align}
		\label{eq:sum of max}
		\max\{a + b,0 \} + \max\{a - b,0 \} \geq 2 \max\{a ,0 \} , 
	\end{align} 
	which can be easily checked by checking the three cases $a \geq |b|$, $ a \le - |b| $ and $ - |b| < a < |b|$. 
	We then calculate 
	\begin{align}
	\label{eq:lhs calculation xz neighbourhood}
	& \sum_{\lambda \in \{0,1\}^k}  \max\bigg\{ \sum_{j= 1}^k (-1)^{\lambda_j} x_j ,0 \bigg\} \\
	\begin{split}
	= & 
	\sum_{\lambda_1,\dots, \lambda_{k-1}\in \{0,1\}} 
	\max\bigg\{x_k+ \sum_{j= 1}^{k-1} (-1)^{\lambda_j} x_j ,0 \bigg\}\\
	& + 
	\sum_{\lambda_1,\dots, \lambda_{k-1}\in \{0,1\}} 
	\max\bigg\{-x_k+ \sum_{j= 1}^{k-1} (-1)^{\lambda_j} x_j ,0 \bigg\} . 
	\end{split}\\
	\ge &\,  2  \sum_{\lambda' \in \{0,1\}^{k-1}} 
	\max\bigg\{\sum_{j= 1}^{k-1} (-1)^{\lambda'_j} x_j ,0 \bigg\}\\
	\stackrel{\text{I.H.}}{\ge } & 2 \cdot 2^{k-2  } |x_1| = 2^{k-1} |x_1|, 
	\end{align}
	where we have used \eqref{eq:sum of max} in the second to last and the induction hypothesis in the last step. 
	This proves the claim. 
\end{proof}
In the proof of Theorem~\ref{thm:complexity of se} we will use that Lemma~\ref{lem:non-stoquasticity 2-local H} implies that every term $a_{i,j} X_iX_j$ contributes an additional cost $\max\{ a_{i,j}, 0 \}$ to the non-stoquasticity of $H_2$. 
Moreover, since $\max_{j\in [k]} |x_j|\ge \sum_{j = 1}^k |x_j|/k$, Lemmas~\ref{lem:non-stoquasticity 2-local H} and \ref{lem:xz-non-stoquasticity} imply that we can view every term $x_{i,j} X_i Z_j$ of $H_2$ as contributing at least a cost $|x_{i,j}|/ (2 \deg(G))$ to the non-stoquasticity of $H_2$, where $G$ is the interaction graph of $H_2$. 

We are now ready to show that with respect to the non-stoquasticity measure $\nu_1$ easing the sign problem for $2$-local XYZ Hamiltonians with on-site Cliffords is \np-complete on arbitrary graphs. 
We restate Theorem~\ref{thm:complexity of se}\ref{item:se i} here. 
\begin{theorem}[\se\, under orthogonal Clifford transformations]
\label{thm:clifford hardness}
	\se\ is \np-complete for $2$-local Hamiltonians on an arbitrary graph, in particular for XYZ Hamiltonians, under on-site orthogonal Clifford transformations, that is, the real group generated by $\{X,Z,W\}$ with $W$ the Hadamard matrix. 
\end{theorem}

\begin{proof}
Clearly the problem is in \np, since one can simply receive a (polynomial-size) description of the transformation in the Yes-case, and then calculate the measure of non-stoquasticity efficiently for XYZ Hamiltonians, verifying the solution. 

To prove \np-hardness, we encode the \maxcut\ problem in the \se\ problem.
A \maxcut\ instance can be phrased in terms of asking whether an anti-ferromagnetic (AF) Ising model on a graph $G = (V,E)$ with $e = |E|$ edges on $v = |V|$ spins
\begin{align}
	H = \sum_{(i,j)\in E} Z_i Z_j , 
\end{align}
has ground-state energy $\lambda_{\text{min}}(H)$ below $A$ or above $B$ with constants $B -  A\geq 1/\poly(v)$. 
This is because in the Ising model one gets energy $-1$ for a $(0,1)$  or $(1,0)$ -edge and $+1$ for a $(0,0)$  or $(1,1)$  edge. 

Let us now encode the \maxcut\ problem phrased in terms of the AF Ising model problem into \se\ for the XYZ Hamiltonian.
We will design a Hamiltonian $H'$, and ask if on-site orthogonal Clifford transformations can decrease its measure of non-stoquasticity $\nu_1$ below $A$, or whether it remains above $B$ for any Clifford basis choice.

For each AF edge between spins $i,j$ in the AF Ising model, the new Hamiltonian $H'$ will have an edge 
\begin{align}
	h_{i,j} = X_i X_j. \label{Hx}
\end{align}
On top of that, for every edge $(i,j) \in E$ we add one ancilla qubit $\xi_{i,j}$ as shown in Figure~\ref{fig:maxcut embedding}, and interactions 
\begin{align}
	h_{i,j}^{(\xi)} = C \left(Z_i Z_j - Z_i Z_{\xi_{i,j}} -   Z_{\xi_{i,j}}Z_j \right) , \label{Hanc}
\end{align}
where $C = 4 \deg(G)$. 
Note that the additional terms are diagonal and hence stoquastic. 
The total new Hamiltonian then reads 
\begin{align}
	H' = \sum_{(i,j) \in E} \left[X_i X_j + C \left(Z_i Z_j - Z_i Z_{\xi_{i,j}} -   Z_{\xi_{i,j}}Z_j \right) \right], \label{eq:embedded maxcut}
\end{align}
and acts on $n = v + e$ qubits. 
We construct $H'$ so that an attempt to decrease the non-stoquasticity $\nu_1$ by swapping $Z$ and $X$ operators via Hadamard transformations will fail, and so the best one can do is to choose a sign in front of each local $X$ operator. Of course, this then becomes the original, hard, \maxcut\ problem in disguise. Let us prove this.

Let $\mc N'(i) = \{ j: (i,j) \in E'\} = \mc N(i) \cup \{ \xi_{i,j}: (i,j) \in E \}$ be the neighbourhood of site $i$ on the augmented graph $G' = (V',E')$ on which $H'$ lives. 
We start the proof by observing that the degree $\deg'(i) = |\mc N'(i)|$ of a site $i$ on the augmented graph satisifies $ \deg'(i) = 2 \deg(i)$, and hence $\deg'(G) =  2\deg(G)$.

\paragraph{Orthogonal Clifford transformations}

First, let us note that any element of the orthogonal single-qubit Clifford group can be written as 
\begin{align}
	C = \pm W^w X^x Z^z , 
\end{align}
where we denote the Hadamard matrix with $W$ and $w,x,z \in \{0,1\}$. 
Since the global sign is irrelevant, a real $n$-qubit Clifford of the form $C = C_1 \otimes \cdots \otimes C_n$ is parametrized by binary vectors $\vec w,\vec x, \vec z \in \{0,1\}^n$. 

How does $H'$ transform under real single-qubit Clifford transformations? 
By definition $C Z C^\dagger \in \{ \pm Z, \pm X\}$ and likewise for $X$. 
Therefore, the transformed Hamiltonian will be of the form \eqref{eq:strictly two-local h}. 
Throughout the proof, we will use that every term $a_{i,j} X_i X_j$ contributes at least $\max\{ a_{i,j},0\}$ to the non-stoquasticity, while every term $x_{i,j} X_i Z_j$ contributes at least $|x_{i,j}| /(2\deg(G')) = |x_{i,j}|/(4 \deg(G))$ as shown by Lemmas~\ref{lem:non-stoquasticity 2-local H} and \ref{lem:xz-non-stoquasticity} above.

We now show that \maxcut\, can be embedded into \se\, under on-site orthogonal Clifford transformations. 
To do so, we need to show two things: 
first, that in the yes-case that $\lambda_{\text{min}}(H) \leq A$, the non-stoquasticity of $H'$ can be brought below $A$ using on-site orthogonal Clifford transformations. 
Second, we show that in the no-case that $\lambda_{\text{min}}(H) \geq B$, the non-stoquasticity of $H'$ cannot be brought below $B$ using on-site orthogonal Clifford transformations.

\paragraph{Yes-case: (Diagonal) transformations that map $X$ to $\pm X$ ($\vec{w} = 0 $).}
These transformations only change the sign in front of $X_i$, keeping its form. 
At the same time they only change the signs of the $Z_i Z_j$ terms, keeping them diagonal and hence stoquastic. 
The transformed $X_i X_j$ terms \eqref{Hx} will be stoquastic if and only if exactly one of the transformations of the $X$ at sites $i,j$ is a $Z$-flip. 

We can view the coefficient $z_i$ as a spin $s_i$ in the original AF Ising model: 
for $z_i = 1$, $X_i \rightarrow -X_i$, corresponding to a spin $s_i = 1$ in the original AF Ising model,
while for $z_i = 0$, $X_i \rightarrow X_i$, which we view it as the Ising spin $s_i=0$.
\begin{center}
	\includegraphics[width = \columnwidth]{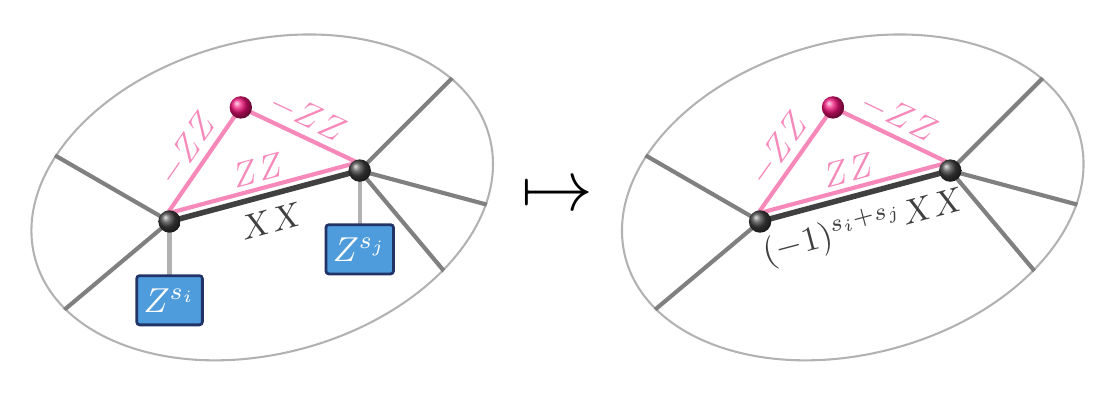}
\end{center}
Each such Clifford transformation thus corresponds to a particular state of the original AF Ising model as given by a spin configuration $\vec s \in \{0,1\}^v$. 
Whenever the transformations on neighbouring sites result in a stoquastic interaction $-X_i X_j$ in the transformed XYZ Hamiltonian, we have a $(0,1)$  or $(1,0)$  anti-ferromagnetic Ising edge with cost $0$. 
On the other hand, each non-stoquastic $X_i X_j$ term in the XYZ Hamiltonian has cost $1$, while the corresponding edge in the Ising model is $(0,0)$  or $(1,1)$  also with cost $1$.

What is the amount of sign easing we can hope to achieve? 
We have argued above that only diagonal transformations which map $X_i \mapsto \pm X_i$ potentially ease the sign problem since we designed the interactions so that a Hadamard transformation always incurs a larger cost than keeping an $X_i X_j$ term non-stoquastic. 
For those transformations, we have a one-to-one correspondence with the ground state of the original AF Ising model. 
Hence, the original AF Ising model ground state energy $\lambda_{\min}(H)$ is also the optimal number of non-stoquastic terms $X_iX_j$ which one can achieve via on-site orthogonal Clifford transformations, each adding an additional cost $1$ to the non-stoquasticity measure $\nu_1$.

In the yes-case we can therefore achieve non-stoquasticity
\begin{align}
	{\nu_1}(\textrm{{\em yes}}) \leq A, 
\end{align}
by choosing $\vec x,\vec w = 0$ and $(z_1, \ldots, z_v)^T =  \vec{s}_0$, the ground state of $H$. 

We now show that in the no-case, the non-stoquasticity measure will be at least
\begin{align}
	{\nu_1}(\textrm{{\em no}}) \geq B.
\end{align}

\paragraph{No-case: (Hadamard) transformations that map $X$ to $\pm Z$ ($\vec{w} \neq 0 $).}
We have designed the additional Hamiltonian term \eqref{Hanc} so that such transformations result in large non-stoquasticity. 
Specifically, we show that for any choice of $\vec{z}$, choosing $\vec{x} = \vec{w} = 0$ achieves the optimal non-stoquasticity in the orbit of orthogonal Clifford transformations.

It is sufficient to show that any Clifford transformation on an edge $(i,j)$ (and its ancilla qubit $\xi_{i,j}$) that is non-stoquastic for a given choice of $\vec{z}$ can only increase the non-stoquasticity. 

To begin with, note that choosing $x_i= 1$ results in $Z_i \mapsto -Z_i$, $X_i \mapsto X_i$, and choosing $w_i =1$ maps $X_i \mapsto Z_i$ and $Z_i \mapsto X_i$.
We obtain the following transformation rules of Pauli $Z_i$ and $X_i$, given choices of $x_i$ and $w_i$: 
\begin{center}
\begin{tabular}{>{\ } cc <{\ }|>{\ }rr <{\quad}}
\toprule
	$x_i$ & $w_i$ & $Z_i$ & $X_i$\\ \midrule
	0 & 0 & $Z_i$ & $X_i$ \\
	0 & 1 & $X_i$ & $Z_i$ \\
	1 & 0 & $-Z_i$ & $X_i$ \\
	1 & 1 & $-X_i$ & $Z_i$ \\\bottomrule
\end{tabular}
\end{center}

First, suppose that for an edge $(i,j)$, a Hadamard transformation is performed on qubit $i$, but not $j$ so that we have $w_i = 1$, $w_j = 0$. 
Then for some choice of $x_i,x_j$ the transformed edge is given by 
\begin{align}
	W_i (X_i X_j \pm C Z_i Z_j )W_i = Z_i X_j \pm C X_i Z_j , 
\end{align}
\begin{center}
	\includegraphics[width = \columnwidth]{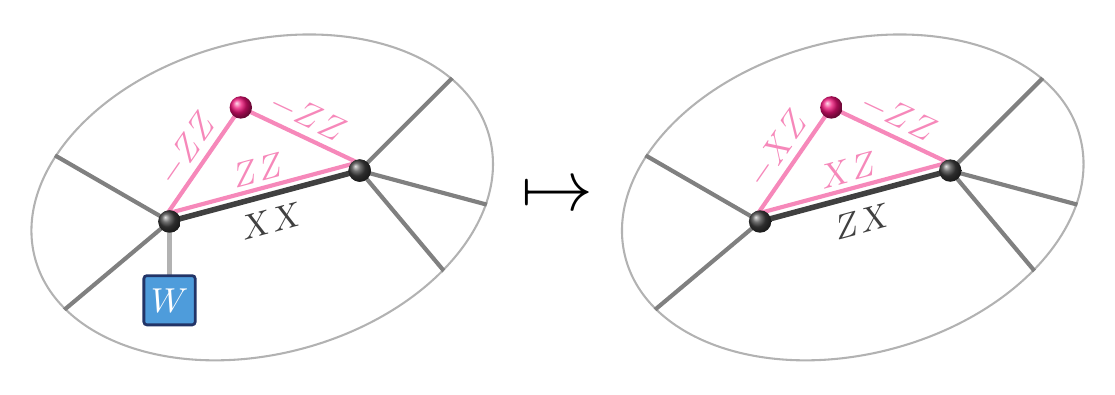}
\end{center}
and has non-stoquasticity at least $ ( C + 1)/(2\deg{G'})$.

Now, suppose that a Hadamard transformation is performed on both qubit $i$ and $j$ but not its ancilla qubit. 
Then the transformed term is given by 
\begin{equation}
\begin{split}
	W_i & W_j  X_i^{x_i} X_j^{x_j} X_{\xi_{i,j}}^{x_{\xi_{i,j}}} ( h_{i,j} + h_{i,j}^{(\xi)})X_i^{x_i} X_j^{x_j} X_{\xi_{i,j}}^{x_{\xi_{i,j}}}  W_i W_j \\
	 & = \pm Z_i Z_j + C \left( \pm X_i X_j \pm X_i Z_{\xi_{i,j}} \pm Z_{\xi_{i,j}}X_j \right), 
\end{split}
\end{equation}
\begin{center}
	\includegraphics[width = \columnwidth]{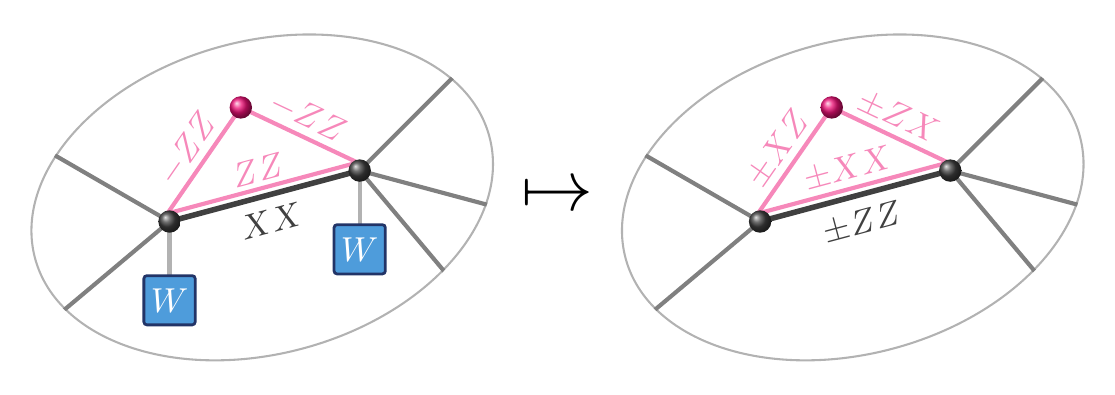}
\end{center}
with non-stoquasticity cost at least 
\begin{align}
	\nu_1\left( C ( \pm X_i Z_{\xi_{i,j}} \pm Z_{\xi_{i,j}}X_j)\right) =  2 C /(2\deg{G'}). 
\end{align}

Could the edge be possibly cured by performing a Hadamard transformation on the ancilla qubit as well? 
In this case, we get 
\begin{align}
\begin{split}
	W_i W_j & W_{\xi_{i,j}}( h_{i,j} + h_{i,j}^{(\xi)}) W_i W_jW_{\xi_{i,j}} \\
	& = Z_i Z_j + C \left(X_i X_j -  X_i X_{\xi_{i,j}} -  X_{\xi_{i,j}}X_j \right), 
\end{split}
\end{align}
\begin{center}
	\includegraphics[width = \columnwidth]{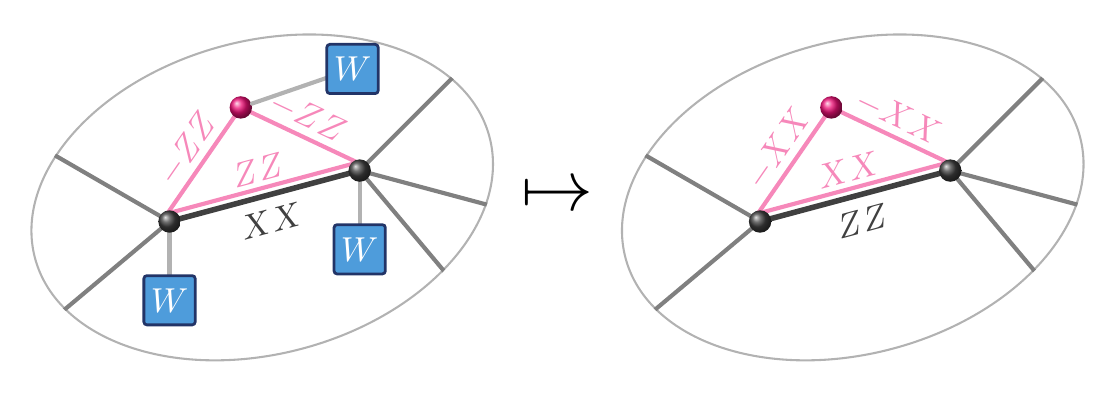}
\end{center}
with non-stoquasticity 
\begin{align}
	\nu_1(W_i W_j W_{\xi_{i,j}}( h_{i,j} + h_{i,j}^{(\xi)}) W_i W_jW_{\xi_{i,j}} ) =  C. 
	\label{eq:edge cost H flip} 
\end{align}
Because of the frustrated arrangement of the signs of the $ZZ $ terms, no local sign flip of those terms (achieved by choices of $x_i,x_j,x_{\xi_{i,j}} \neq 0$) can cure the sign problem of an ancillary triangle, leaving it lower bounded by \eqref{eq:edge cost H flip}.

On the other hand, the original cost incurred from local sign flips via $Z$-transformations, is given by 
\begin{align}
	\nu_1(X_i X_j) = 1 , 
	\label{eq: maximum edge cost}
\end{align}
which is always smaller than the cost incurred if additional $X$ or $W$ transformations are applied since we chose $C$ such that $C /(2\deg{G'})  = 1$. 
Therefore, in the no-case of the original AF Ising model the non-stoquasticity of $H'$ cannot be brought below 
\begin{align}
	{\nu_1}(\textrm{{\em no}}) \geq B, 
\end{align}
with the optimal choice achieved for $\vec x,\vec w = 0$ and $(z_1, \ldots, z_v)^T = \vec{s}_0$. 

\end{proof}

\begin{theorem}[\se\ under orthogonal transformations]
\label{thm:orthogonal hardness}
	\se\ is \np-complete for $2$-local Hamiltonians on an arbitrary graph, in particular, for XYZ Hamiltonians under on-site orthogonal transformations. 
\end{theorem}
\begin{proof}
	We proceed analogously to the proof for orthogonal Clifford transformations, showing that in the yes-case, there exists a product orthogonal transformation $O = O_1 \cdots O_n$ such that $\nu_1(O H' O^T) \leq A$, while in the no-case there exists no such transformation $P$ with $\nu_1(P H' P^T) \leq B$. 

	The yes-case is clear: In this case, the energy of $\vec s$ is below $A$. 
	Then by our construction, the sign problem can be eased below $A$ with the orthogonal transformation 
	\begin{align}
		O = \prod_{i \in V} Z_i^{s_i} . 
		\label{eq:optimal Z trafo}
	\end{align}

	We now need to show that in the no-case, any orthogonal transformation incurs non-stoquasticity above $B$.
	We first remark that the orthogonal group $O(2)$ decomposes into two sectors with determinant $\pm 1$, respectively. 
	Therefore, any $2 \times 2$ orthogonal matrix can be written as
	\begin{align}
			O_a(\theta)  =
		\begin{pmatrix} \cos \theta & -a \sin \theta\\ \sin \theta & a \cos \theta
		\end{pmatrix} , 
	\end{align}
	which for $a = \det(O_a(\theta)) = -1$ is a reflection and for $a = \det(O_a(\theta)) = +1 $ a rotation by an angle $\theta$.
	Note that the following composition laws hold 
	\begin{align}
	 	O_{-1}(\theta)O_{1}(\phi)&  = O_{-1}(\theta- \phi), \\
	 	O_{1}(\theta)O_{-1}(\phi) & = O_{-1}(\theta + \phi), \\
	 	O_1(\theta)O_1(\phi)&  = O_{1}(\theta +  \phi), \\
	 	O_{-1}(\theta)O_{-1}(\phi)&  = O_{1}(\theta - \phi). 
	\end{align} 
	 
	Now observe three facts: 
	First, any reflection by an angle $\theta$ can be written as a product of a reflection across the $X$-axis and a rotation as
	\begin{align}
		\begin{pmatrix} \cos \theta &  \sin \theta\\ \sin \theta & - \cos \theta
		\end{pmatrix} = \begin{pmatrix} \cos \theta &  - \sin \theta\\ \sin \theta & \cos \theta
		\end{pmatrix} Z = R(\theta) Z ,   
	\end{align}
	where $R(\theta)$ is the rotation by an angle $\theta$. 
	Second, for any Hermitian matrix $H$ and any angle $\theta$, it holds that 
	\begin{align}
		O(\theta) H O(\theta)^T = O(\theta + \pi) H O(\theta + \pi)^T , 
	\end{align}
	so that it suffices to restrict to angles $\theta \in [-\pi/4, 3\pi/4]$ in an interval of length $\pi$. 
	Third, a rotation by an angle $\pi/2$ can be decomposed into two reflections as 
	\begin{align}
		R(\pi/2) = X Z .  
	\end{align}
	Taken together, these facts imply that an arbitrary single-qubit orthogonal transformation is given by 
	\begin{align}
		\label{eq:orthogonal standard form}
		O(\theta,z,p) =  R\left(\theta + (\pi/2)^p\right) \cdot Z^z = R(\theta) \cdot X^p Z^{p + z } , 
	\end{align}
	where the rotation angle is given by $\theta \in [ -\pi/4, \pi/4]$, $z \in \{ 0,1\}$ fixes whether or not a $Z$-flip is applied, and $p \in \{0,1\}$ mods out a rotation by an angle $\pi/2$. 
	Now define $O(\vec\theta,\vec z, \vec p ) = \prod_i O_i(\theta_i,z_i,p_i)$.

	We now need to show that, in the no-case, for any choice of $\vec \theta \in [-\pi/4,\pi/4]^n$, $\vec z, \vec p \in \{0,1\}^n$ it holds that 
	\begin{align}
		\nu_1 \left(O(\vec \theta, \vec z, \vec p )H'O(\vec \theta, \vec z, \vec p )^T\right) \geq b . 
	\end{align}
	To complete the proof, we use that the action of $R(\theta)$ on Pauli-$X$ and $Z$ matrices is given by 
	\begin{align}
		R(\theta) Z R(\theta )^T & = \cos(2\theta) Z + \sin(2\theta) X,
		\label{eq:orthogonal effect Z}\\
		R(\theta) X R(\theta )^T & = \cos(2\theta) X - \sin(2\theta) Z . 
		\label{eq:orthogonal effect X}
	\end{align}
	In Appendix \ref{sec:orthogonal calculations}, we show that given any choice of $\vec z$ and $\vec p$, no choice of $\vec \theta$ can decrease the non-stoquasticity of an edge with non-zero non-stoquasticity below $1$.
	That is, we analyze transformations of the following form: 
	\begin{center}
	\includegraphics[width = .7\columnwidth]{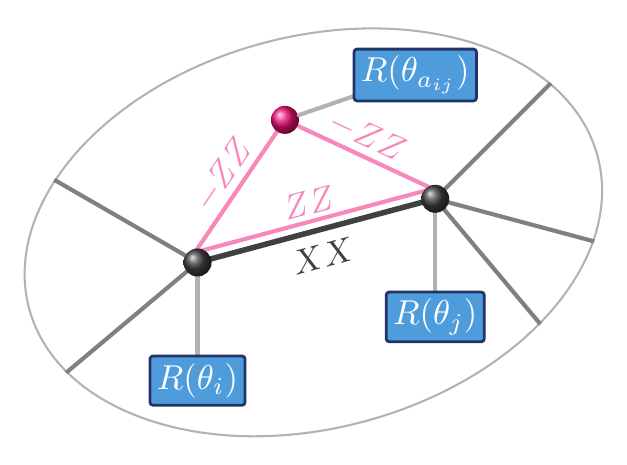} 
	\end{center}
	We do this by using the standard form \eqref{eq:orthogonal standard form} of the local transformations in terms of $X$, $Z$ and restricted rotation matrices with angles in $[- \pi/4, \pi/4 ] $. 
	We split the proof into three parts, analyzing three different (continous) regions for rotation angle choices.
	The only difference in the construction when compared to the Clifford-case is that here we choose $C =  (2 \deg G')^2$. 
	This completes the proof. 
\end{proof}

Notice that since \maxcut\ is \np-hard already on subgraphs of the double-layered square lattice \cite{barahona_computational_1982}, which has constant degree, there are hard instances of the \se\ problem on low-dimensional lattices with constant coupling strength. 

We remark that one can easily extend the proofs of Theorems~\ref{thm:clifford hardness} and \ref{thm:orthogonal hardness} for different non-stoquasticity measures $\nu_p$ with $1< p < \infty$. 
To see this, note that the decision problem for $\nu_p(H)$ is equivalent to the problem for $\nu_p(H)^p$. 

For terms $x_{i,j}X_i Z_j$ we can then use the trivial bound 
\begin{align}
\begin{split}		
	\sum_{\lambda \in \{0,1\}^k} \max\bigg\{ \sum_{j= 1}^k (-1)^{\lambda_j} x_{i,j} ,0 \bigg\}^p \\
	\geq 2^{p (k-\deg(i))}  \sum_j|x_j|^p  
	\label{eq:nonstoq arbitrary xz vertex p norm}
\end{split}
\end{align}
instead of Lemma~\ref{lem:xz-non-stoquasticity}. 
Thus, every term $x_{i,j}X_i Z_j$ contributes at least $2^{ -p\deg G'}  |x_{i,j}|^p$ to $\nu_p^p$, while a term $a_{i,j}X_iX_j$ contributes
\begin{align}
	\nu_p(a_{i,j}X_i X_j)^p  = (\max\{a_{i,j},0\})^p , 
\end{align}
to the total non-stoquasticity of $H'$. 

For general $\ell_p$-non-stoquasticity measures $\nu_p$ one therefore need merely choose $C = 2^{\deg (G')} = 4^{\deg (G)}$ to prove Theorem~\ref{thm:clifford hardness} and $C = 2^{2 \deg (G')} = 16^{\deg (G) }$ for Theorem~\ref{thm:orthogonal hardness}.

\ifarxiv
\section*{Acknowledgements}
We are immensely grateful for the many fruitful discussions which have helped shape this work -- with Albert 
Werner, Martin Schwarz, Juani Bermejo-Vega, 
and Christian Krumnow in early stages of the project; more recently with Matthias Troyer, Joel Klassen, Marios Ioannou, Maria Laura Baez, Hakop Pashayan, Simon Trebst, Augustine Kshetrimayum, Alexander Studt and Alex Nietner. 
We also thank Barbara Terhal, Marios Ioannou, Jarrod McClean and Maria Laura Baez for helpful comments on our draft. 
D.\,H., I.\,R.\ and J.\,E.\ acknowledge financial support 
from the ERC (TAQ), 
the Templeton Foundation, the DFG 
(EI 519/14-1, EI 519/9-1, EI 519/7-1, CRC 183 in project B01), 
and the European Union's Horizon 2020 research and innovation programme under grant agreement No 817482 (PASQuanS). D.\,N.\ 
has received funding from the People Programme (Marie Curie Actions) EU’s 7th Framework Programme under REA grant agreement No. 609427. His research has been further co-funded by the Slovak Academy of Sciences, 
as well as by the Slovak
Research and Development Agency grant QETWORK APVV-14-0878 and VEGA MAXAP 2/0173/17.

\bibliographystyle{./myapsrev4-1}
\bibliography{signs}
\fi


\widetext
\begin{appendix}


\section{Conjugate gradient descent for sign easing}
\label{app:conjugate gradient}

In this appendix, we provide details on the numerical implementation \cite{optimization_numerics} of the conjugate gradient descent algorithm for minimizing the non-stoquasticity $\nu_1$ over orthogonal circuits. 
In this work, we focus on the ansatz class of translation-invariant on-site orthogonal transformations, but other classes such as constant-depth circuits can be implemented analogously. 

\subsection{Translation-invariant formulation of the non-stoquasticity measure}

We begin by deriving a simple expression for the non-stoquasticity measure for translation-invariant, one-dimensional, nearest-neighbour Hamiltonians.  
Our formulation is based on the observation that by translation-invariance, the measure $\nu_1(H)$ only depends on the local Hamiltonian term $h$ and is therefore independent of the system size. 
Since Hamiltonian terms with overlapping support can contribute to the same matrix element in the global Hamiltonian matrix it is thus sufficient to optimize certain sums of matrix elements of $h$ rather than make $h$ itself stoquastic -- a condition required only by the stronger notion of term-wise sxtoquasticity. 

Recall that we consider real nearest-neighbour Hamiltonians with closed boundary conditions on a line of $n$ local systems with dimension $d$
\begin{align}
	H = \sum_{i =1 }^n T_i(h) , 
\end{align}
where $h \in \mb R^{d^2 \times d^2}$ is the local term and we use the notation $T_i(h) = \id^{\otimes (i -1)} \otimes h\otimes \id^{\otimes (n - i - 1)}$. 
Closed boundary conditions identify $n+1 = 1$. 

Specifically, we can calculate 
\begin{align}
\label{eq:effective non-stoquasticity}
 \nu_1(H) = n 2^{n - 3} \sum_{\substack{i,j,k;l,m,n: \\k\neq l, m = n}} \max \left\{ 0,\left(h \otimes \id + \id \otimes h  \right)_{i,k,m;j,l,n}\right\} \equiv n 2^{n-3} \tilde \nu_1(h). 
\end{align}
This reduces the problem of minimizing $\nu_1(H)$ over on-site orthogonal bases to minimizing $\tilde \nu_1(h)$ over such bases. 
To derive Eq.~\eqref{eq:effective non-stoquasticity} we re-express the global non-stoquasticity measure as follows:  
\begin{align}
	\nu_1(H)&  = \sum_{\substack{(i_1, \ldots, i_n)\neq\\(j_1, \ldots,j_n)}}\max\left\{ 0, \bra{i_1, \ldots, i_n} \sum_{l=1}^n T_l(h) \ket{j_1, \ldots, j_n}\right\}\\
	 & =  \sum_{p = 1}^n \sum_{\substack{i_1, \ldots, i_n,\\j_p, j_{p+1}: \\ i_{p+1}\neq j_{p+1} }}
	 \max\left\{ 0, \bra{i_1, \ldots, i_n} (T_p(h) + T_{p+1}(h) )
	 \ket{i_1, \ldots, i_{p-1},j_p,  j_{p + 1},i_{p+2}, \ldots, i_n}\right\}\\
 &= 2^{n -3} \sum_{p = 1}^n \sum_{\substack{i_p, \ldots, i_{p + 2},\\j_p,j_{p+1}: \\
 i_{p+1}\neq j_{p+1}
 }}
 \max\left\{ 0, \bra{i_p, \ldots, i_{p + 2}} 
 (h \otimes \id + \id \otimes h)
 \ket{j_p, j_{p + 1},i_{p+2}}\right\}
 \\
  &= n 2^{n -3} \sum_{\substack{i_1, \ldots, i_{3},\\j_1, j_{2}: \\
 i_{2}\neq j_{2} }}\max\left\{ 0, \bra{i_1, i_2, i_3} (h \otimes \id + \id \otimes h ) \ket{j_1, j_2, i_3} \right\} . 
 \end{align}
In the first step, we have used that the condition ($i_1, \ldots, i_n)\neq (j_1, \ldots,j_n)$ implies that at least one of the indices differs. 
Summing over $p $, we can let this index be the $(p +1)^\text{st}$ one. 
To avoid double-counting, we then divide the sum over all strings which differ on at most two nearest neighbours into a sum over all strings which potentially differ on two nearest neighbours \emph{left} of the $(p+1)^\text{st}$ index. 
Patches which differ on more than two nearest neighbours vanish for nearest-neighbour Hamiltonians. 
In the second step, we use that all terms with support left of the $p^\text{th}$ qubit vanish, since the basis strings differ on the $(p +1)^\text{st}$ qubit.
In the last step, we use the translation-invariance again to account for the sum over $p$, incurring a factor $n$.

\subsection{Gradient of the objective function}

As an ansatz class we choose on-site orthogonal transformations in $O(d)$, where $d$ is the dimension of a local constituent of the system. 
More precisely, we consider transformations of the type (see Eq.~\eqref{eq:local basis choice}) 
\begin{align}
	H = \sum_{i = 1}^{n} T_i(h) \mapsto O^{\otimes n} H (O^T)^{\otimes n} ,
\end{align}
which locally amounts to 
\begin{align}
	h \mapsto h(O) \coloneqq (O \otimes O) h (O^T \otimes O^T) . 
\end{align}

The key ingredient for the conjugate gradient descent algorithm is the derivative of the objective function $ \tilde \nu_1(h(O))$ with respect to the orthogonal matrix $O$. 
The gradient is given by 
\begin{align}
	\label{eq:global gradient}
	\frac{\partial}{\partial O} \tilde\nu_1(h(O))  = \sum_{i,j} \frac{\partial \tilde \nu_1}{\partial h(i|j)} \cdot \frac{\partial h(i|j)}{\partial O} .
\end{align}
We can further expand the terms of the global gradient \eqref{eq:global gradient}: 
in particular, we can express the gradient of the effective local terms as a conjugation $ h( O) = \mathcal C( O) h \mathcal C^\dagger( O)$ of the local term $h$ by the orthorgonal circuit $\mathcal C( O)  = O \otimes O $. 
We will now  derive expressions for the measure and the gradients that the algorithm has to evaluate.

\paragraph{The objective function gradient}

We first determine the gradient of the objective function. 
Since we will also make use of different measures $\nu_p$ as defined in Eq.~\eqref{eq:nu_p}, we write the objective function for different $p \geq 1$ as 
\begin{align}
	\tilde\nu_p^p(h) = \sum_{\substack{i,j,k;l,m,n: \\k\neq l, m = n}} \max \left\{ \left(h \otimes \id + \id \otimes h  \right)_{i,k,m;j,l,n}  ,0 \right\}^p. 
\end{align}
Then, the gradient of the objective function is given by 
\begin{align}
 \left. \frac{\partial \tilde \nu_p^p}{\partial h(x|y)}\right \vert_{h(0)} =p  \sum_{\substack{i,j,k;l,m,n: \\k\neq l, m = n}}\left( \ket{x}\bra{y} \otimes \id + \id \otimes \ket{x} \bra {y} \right)_{i,k,m;j,l,n}   \max \left\{ \left(h \otimes \id + \id \otimes h  \right)_{i,k,m;j,l,n}  ,0 \right\}^{p-1}
\end{align}

\paragraph{Gradient of the transformed Hamiltonian}
We can expand the gradient of the transformed Hamiltonian term by the orthogonal matrix as 
\begin{equation}
	\frac{\partial h(i|j)}{\partial O} 
		= 
		\sum_{m,n} 
		\frac{
			\partial\operatorname{adj}_h(\mathcal C)(i | j)
		}{
			\partial \mathcal C(m | n)
		}
		\frac{
			\partial\mathcal C(m|n)
		}{
			\partial O ,  
		}
\end{equation}
where by $\operatorname{adj}_h(\mathcal C) $ we denote the adjunction map $h \mapsto \mc C h \mc C^T$. 
The derivative of the adjoint action of $\mc C$ on $h$ is given by 
\begin{align}
	\frac{
		\partial \operatorname{adj}_h(\mathcal C)(i | j)
	}{
		\partial \mathcal C(m|n)
	} & = 
		\sum_{k,l}\frac\partial{\partial \mathcal C(m|n)} \mathcal C(i | k) h(k|l) \mathcal C(j|l) \\
		&\hspace{-2cm}= \sum_{k,l}\left[\delta_{m,i} \delta_{n,k} h(k|l) \mathcal C(j|l) + \mathcal C(i|k)h(k|l) \delta_{m,j}\delta_{n,l}\right] \\
		&\hspace{-2cm}= \braket m i \bra j\mathcal C h^T\ket n + \braket m j \bra i\mathcal C h \ket n 
\end{align}

From this expression, we can directly read off its matrix form 
\begin{align}
	\frac{
		\partial \operatorname{ad}_h(\mathcal C)(i | j)
	}{
		\partial \mathcal C
	}	= \ket i \bra j \mathcal C h^T + \ket j \bra i  \mathcal C h. 
\end{align}

It remains to compute the gradient of the circuit with respect to the orthogonal matrix. 
We obtain with $\bra m = \bra {m_1}\otimes \bra{m_2}$
\begin{align}
	\frac{\partial\, \mathcal C(m|n) } {\partial O(k|l)} = \sum_{i = 1}^{2}  \delta_{m_i,k} \delta_{n_i,l} O(m_1|n_1)O(m_2|n_2) = \sum_{i = 1}^{2} \braket{m_{i}}{k} \braket{l}{n_i} O(m_1|n_1) O(m_2|n_2),  
\end{align}
which expressed in matrix form is then given by 
\begin{align}		
	\frac{\partial\, {\mathcal C}_{l_e}(m|n) } {\partial O} 	=\ket{n_1}\bra{m_1} \bra{m_2} O \ket{n_2}  + \ket{n_2}\bra{m_2} \bra{m_1} O \ket{n_1} . 
\end{align}

\subsection{Algorithmic procedure}
\label{sec:algorithmic procedure}

We start our conjugate gradient algorithm either at the identity matrix (with or without a small perturbation) or a Haar-randomly chosen orthogonal matrix as indicated at the respective places in the main text. 

Since the minimization of the $\tilde \nu_1$-measure is numerically not well behaved, we improve the performance of the algorithm in several ways. 
First, we observe that the measure $\tilde \nu_2$ given by the Frobenius norm of the non-stoquastic part of the Hamiltonian is numerically much better behaved. 
This is due to the $\ell_2$-norm being continuously differentiable while the $\ell_1$-norm is only subdifferentible. In particular, at its minima the gradient of the $\ell_1$-norm is discontinuous and never vanishes. 
For this reason, rather than optimizing the $\ell_1$-norm of the non-stoquastic part, we optimize a smooth approximation thereof as given by~\cite{schmidt_fast_2007}
\begin{align}
	\label{eq:smooth norm}
	\nu_{1,\alpha}(H) \coloneqq \sum_{i\neq j} f_\alpha\left( H_{i,j} \right) , 
\end{align}
with $f_\alpha(x) =  x + \alpha^{-1} \log(1 + \exp(-\alpha x))$. 

To achieve the best possible performance, we then carry out a hybrid approach: 
First, we pre-optimize by minimizing $\tilde \nu_2$ using our conjugate gradient descent algorithm. 
Second, starting at the minimizer obtained in the Frobenius norm optimization, we minimize the smooth non-stoquasticity measure $\tilde \nu_{1,\alpha}$. We choose the value $\alpha = 50$, $\alpha = 100$ and $\alpha = 40$, for the random stoquastic Hamiltonians, the \jmodel-model, and the frustrated ladder model, respectively. 
We then compare the result to a direct minimization of the non-stoquasticity $\tilde \nu_{1,\alpha}$ starting from the original initial point and choose whichever of the minimizations 
performed best.
The exact details of our optimization algorithms can be found at~Ref.~\cite{optimization_numerics} together with code to reproduce the figures shown here. 
Our code framework can be easily adapted for optimization of other Hamiltonian models and more general circuit architectures. 


\section{Orthogonal transformations of the penalty terms (proof of Theorem~\ref{thm:orthogonal hardness})}
\label{sec:orthogonal calculations}

\begin{proof}[Proof of Theorem~\ref{thm:orthogonal hardness} (continued)]

As in the proof for orthogonal Clifford transformations, we will show that for any given choice of $Z$-transformations one cannot further decrease the non-stoquasticity by exploiting the additional freedom offered by the full orthogonal group. 
Above, we have argued that applying an arbitrary orthogonal transformation at a single site can be reduced to applying $R(\theta) X^p Z^{p + z}$ with $\theta \in [- \pi/4,\pi/4]$ and $z,p \in \{0,1\}$. 
We will now show that for any choice of $\vec z$ and $\vec p$, a rotation by angles $\vec \theta \in [- \pi/4,\pi/4]^n$ cannot decrease the non-stoquasticity any further. 

Analogously to the proof for Clifford-transformations, we discuss all possible transformations by dividing them into different cases. In each case the non-stoquasticity of an uncured edge $(i,j)$ and its ancilla qubit $\xi_{i,j}$ cannot be eased below its previous value of $1$. 
The additional difficulty we encouter here is that the orthogonal group is continuous as opposed to the orthogonal Clifford group, which is a discrete and rather `small' group. 

Given a choice of $\vec z$, consider an edge $(i,j)$ with a non-trivial contribution to $\nu_1$ and its corresponding ancilla qubit $\xi_{i,j}$.
We begin, supposing that $p_i = p_j = p_{\xi_{i,j}} = 0 $ so that the $X$-flips act trivially on all three qubits. 
We now analyze the effect of the remaining rotations $R(\theta)$ on each of the qubits. 
\begin{center}
	\includegraphics{fig_orthogonal_trafo.pdf}
\end{center}

Specifically, we apply rotations with angles $\theta_i/2$, $\theta_j/2$, $\theta_{\xi_{i,j}}/2$ with  $\theta_i, \theta_j, \theta_{\xi_{i,j}} \in [-\pi/2, \pi/2]$ to the three qubits.
Note that we consider rotations by half-angles $\theta \rightarrow \theta/2$ while at the same time doubling the interval $[-\pi/4, \pi/4] \rightarrow [-\pi/2, \pi/2]$ to ease notation later in the proof. 
The effect of rotations on each vertex of an edge $(i,j)$ is given by 
\begin{equation}
\begin{split}
X_i X_j & + C Z_i Z_j \mapsto \\& [C \cos(\theta_i) \cos(\theta_j)- \sin(\theta_i) \sin(\theta_j)] Z_i Z_j \\
& + [C \cos(\theta_i) \sin(\theta_j) -  \sin(\theta_i) \cos(\theta_j)] Z_i X_j \\
& + [C \sin(\theta_i) \cos(\theta_j) - \cos(\theta_i) \sin(\theta_j)] X_i Z_j \\
& + [C \sin(\theta_i) \sin(\theta_j) + \cos(\theta_i) \cos(\theta_j)] X_i X_j ,  
\end{split}
\end{equation}
and likewise for the edges $(i,\xi_{i,j})$ and $(j,\xi_{i,j})$. 
The non-stoquasticity of the transformed  Hamiltonian terms corresponding to the three qubits is then given by  
\begin{align}
	\nu_1
	\bigg( R_i&(\theta_i/2)  R_j(\theta_j/2) R_{\xi_{i,j}}(\theta_{\xi_{i,j}}/2) 
		\left(X_i X_j  + C ( Z_i Z_j  - Z_i Z_{\xi_{i,j}}- Z_j Z_{\xi_{i,j}}) \right)
	R_i(\theta_i/2)^T R_j(\theta_j/2)^T R_{\xi_{i,j}}(\theta_{\xi_{i,j}}/2)^T \bigg) \\
	\geq & \, \max \left\{C \sin(\theta_i) \sin(\theta_j) + \cos(\theta_i) \cos(\theta_j), 0 \right\} \label{eq:i}\\
	& + (2\deg{G'})^{-1} \cdot
	\bigg( 
		|C \sin(\theta_i) \cos(\theta_j) - \cos(\theta_i) \sin(\theta_j)|
		+ |C \sin(\theta_j) \cos(\theta_i) - \cos(\theta_j) \sin(\theta_i)| 
	\bigg) \label{eq:ii}  \\
	& + \bigg(\max \left\{-C \sin(\theta_i) \sin(\theta_{\xi_{i,j}}), 0 \right\}
	 + \max \left\{-C \sin(\theta_j) \sin(\theta_{\xi_{i,j}}), 0 \right\} 
	\bigg) \label{eq:iii}\\
	& + (2\deg{G'})^{-1} \cdot C 
	\bigg( 
		|\cos(\theta_i) \sin(\theta_{\xi_{i,j}})| 
		+ |\sin(\theta_i) \cos(\theta_{\xi_{i,j}})| 
		+ |\cos(\theta_j) \sin(\theta_{\xi_{i,j}})| 
		+ |\sin(\theta_j) \cos(\theta_{\xi_{i,j}})| 
	\bigg) \label{eq:iv} . 
\end{align}
Note that the terms \eqref{eq:i} and \eqref{eq:iii} stem from the $XX$ interactions with a positive sign and therefore depend on the signs of $\sin$ terms. 
Conversely, the terms \eqref{eq:ii} and \eqref{eq:iv} stem from the $XZ$ interactions and therefore involve absolute values. 

\begin{figure}
	\centering
	\includegraphics{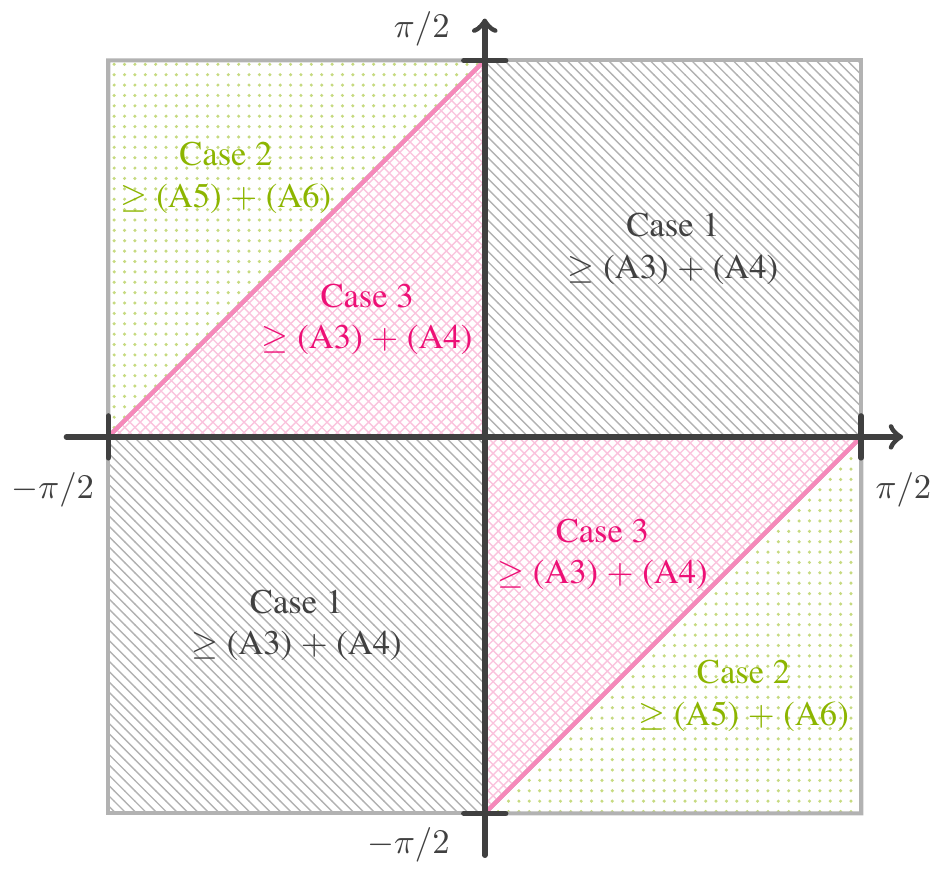}
	\caption{We divide the derivation of a lower bound on the non-stoquasticity of an edge $(i,j)$ and its ancilla qubit $\xi_{i,j}$ on which we apply rotations $R_i(\theta_i/2) R_j(\theta_j/2) R_{\xi_{i,j}}(\theta_{\xi_{i,j}}/2) $ with $\theta_i,\theta_j,\theta_{\xi_{i,j}} \in [\pi/2,\pi/2]$ into three distinct cases.
	In each case we use different terms of the expression \eqref{eq:i}-\eqref{eq:iv} to lower bound the non-stoquasticity.  
	\label{fig:cases orthogonal proof} 
	}	
\end{figure}

We divide the allowed rotations into different sectors corresponding to the different combinations of the signs of $\sin(\theta_i)$ and $\sin(\theta_j)$ as shown in Fig.~\ref{fig:cases orthogonal proof}.
Which of the terms in \eqref{eq:i} and \eqref{eq:iii} are non-trivial precisely depends on these combinations. 
We divide the cases as follows:
 first, the sectors in which $\sign(\theta_i) = \sign(\theta_j)$. 
 Second, the sectors in which $\sign(\theta_i) = - \sign(\theta_j)$ and $|\theta_i | + |\theta_j | \leq \pi/2$.
Third, the sectors in which $\sign(\theta_i) = - \sign(\theta_j)$ and $|\theta_i | + |\theta_j | \geq \pi/2$.
Taken together, the three cases cover the entire range of allowed angles. 
Moreover, in all cases we allow arbitrary choices of $\theta_{\xi_{i,j}} \in [-\pi/2,\pi/2]$.
We now proceed to lower-bound the non-stoquasticity of the Hamiltonian terms acting on the three qubits, where in each case we will use different terms of Eqs.~\eqref{eq:i}-\eqref{eq:iv}.

We first discuss the case in which $ \sign(\theta_i) = \sign(\theta_j)$. 
In this case, it suffices to consider the terms \eqref{eq:i} and \eqref{eq:ii}. 
Observe that in this case both $C \sin(\theta_i) \sin(\theta_j) \ge 0 $ and $\cos(\theta_i)\cos(\theta_j) \geq 0$. 
Moreover, to our choice of $C$ and noting that 
$(|C \sin(\theta_i) \cos(\theta_j) - \cos(\theta_i) \sin(\theta_j)|
+ |C \sin(\theta_j) \cos(\theta_i) - \cos(\theta_j) \sin(\theta_i)| )\geq C |\sin(\theta_i - \theta_j)| $, we obtain 
\begin{align}
	\eqref{eq:i} + \eqref{eq:ii} \geq  \cos(\theta_i -  \theta_j )
	+ C /(2 \deg{G'}) \cdot |\sin(\theta_i -  \theta_j )| \geq  1 .
	\label{eq:bound equal signs}
\end{align}

Second, we discuss the case in which $\sign(\theta_i) = - \sign(\theta_j)$ and $\pi/2 < |\theta_i| + |\theta_j| \leq \pi$.
In this case, we consider the terms \eqref{eq:iii} and \eqref{eq:iv}: 
\begin{align}
	\eqref{eq:iii} +  \eqref{eq:iv}  \geq  & \,   C /(2 \deg{G'}) \cdot  \bigg( |\cos(\theta_{\xi_{i,j}})| ( \sin(|\theta_i|) + \sin(|\theta_j|)) \\
	& +  |\sin(\theta_{\xi_{i,j}})| \left( 2 \deg(G')  \min \{ \sin(|\theta_i|), \sin(|\theta_j|)\} +  |\cos(\theta_i)| + |\cos(\theta_j)| \right) \bigg) \\
	& \geq 1 . 
\end{align}
Here, we have used that 
\begin{align}
	\sin(|\theta_i|) + \sin(|\theta_j|) = 2 \sin \left( \frac{ | x | + | y | } { 2} \right) \cos \left( \frac{ | x | - | y | } { 2} \right) \geq 1, 
\end{align}
for $\pi/2 \leq |\theta_i| + |\theta_j| \leq \pi$ so that $|x| - |y| \leq \pi/2$ and the definition of $C =  (2\deg{G'})^2 $. 

\begin{figure}[t]
	\centering
	\includegraphics{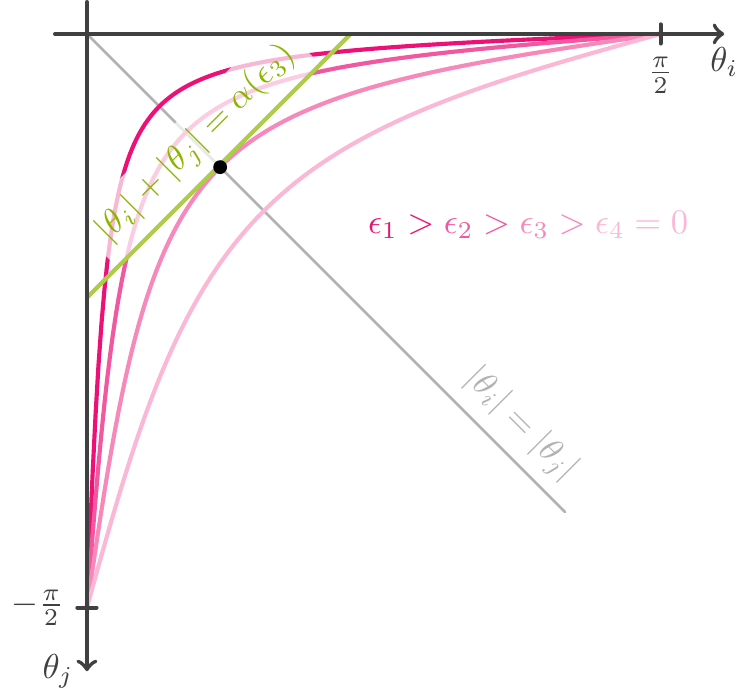}
	\caption{Illustration of the lower bound in case 3: Every $\theta_i$ and $\theta_j$ such that $|\theta_i| + |\theta_j| \leq \pi/2$ defines a contour line $\epsilon(\theta_i, \theta_j) = \epsilon_1, \epsilon_2, \epsilon_3, \epsilon_4 = 0$ (shades of pink). 
	Given $\theta_i, \theta_j$ and defining $\overline \epsilon = \epsilon(\theta_i, \theta_j)$, Lemma ~\ref{lem:minimal value} implies the lower bound $|\theta_i| + |\theta_j| \geq \alpha(\overline \epsilon)$ as defined in Eq.~\eqref{def:alpha(epsilon)}. 
	We then obtain $\eqref{eq:i} + \eqref{eq:ii} \ge  \overline \epsilon + (C+1) \sin(\alpha(\overline \epsilon)) \geq 1$. 
	\label{fig:minimize arctan}
	}
\end{figure}

Finally, we have the remaining case $\sign(\theta_i)= -  \sign(\theta_j) $ and $|\theta_i|+ |\theta_j| \leq \pi/2$. 
In this case, it is again sufficient to consider terms \eqref{eq:i} and \eqref{eq:ii}. 
This is the hardest case since the $\sin$ terms in \eqref{eq:i} increase much faster than the $\cos$ terms decrease due to the factor of $C > 1$. 
Therefore, we cannot find a bound in terms of a sum-of-angles rule as in the previous cases. 
Instead, to
lower-bound the terms terms \eqref{eq:i} and \eqref{eq:ii} in this case, we proceed as follows: 
We reduce the problem of minimizing the sum $\eqref{eq:i} + \eqref{eq:ii} \geq 1$ to a one-dimensional problem by noting two facts: 
first, the term \eqref{eq:ii} depends only on the sum $|\theta_i| + |\theta_j|$. 
Moreover, it increases monotonously in this sum.
Second, for every choice of $\theta_i, \theta_j$, the value of \eqref{eq:i} takes on its minimal value $\overline \epsilon$ 
at $|\theta_i| = |\theta_j| \eqqcolon \alpha(\overline \epsilon)/2$. 
Hence, the sum $\eqref{eq:i} + \eqref{eq:ii}$ is lower bounded by the sum of $\overline \epsilon$ and \eqref{eq:ii} evaluated at $\alpha(\overline \epsilon)$.
The proof is concluded by a lower bound on the latter term.
We now elaborate those steps one-by-one.

Let us begin by expressing \eqref{eq:ii} as a function of $|\theta_i| + |\theta_j|$ 
\begin{align}
		 & |C \sin(\theta_i) \cos(\theta_j) - \cos(\theta_i) \sin(\theta_j)|
		+ |C \sin(\theta_j) \cos(\theta_i) - \cos(\theta_j) \sin(\theta_i)| \\
		& = C \sin(|\theta_i|) \cos(\theta_j) + \cos(\theta_i) \sin(|\theta_j|)
		+ C \sin(|\theta_j|) \cos(\theta_i) + \cos(\theta_j) \sin(|\theta_i|) \\
		& = (C + 1) \sin (|\theta_i| + |\theta_j|) , 
\end{align}
where we have used the fact that $\sign(\sin(\theta_i)) = - \sign(\sin(\theta_j))$ and that the cosines are non-negative. 
For $|\theta_i| + |\theta_j| \leq \pi/2$ this is a monotonously increasing function in $|\theta_i| + |\theta_j|$. 
We now define the value of the term \eqref{eq:i} to be 
\begin{align}
	\epsilon(\theta_i, \theta_j) \coloneqq \max\{ - C \sin|\theta_i| \sin |\theta_j| + \cos |\theta_i| \cos |\theta_j|,0\}  \geq 0 . 
	\label{eq:condition i}
\end{align}

For every choice of $\alpha = \alpha(\theta_i, \theta_j) \coloneqq |\theta_i| + |\theta_j|$, the minimal value of 
\begin{align}
\label{eq:minimize}
	\eqref{eq:i} + \eqref{eq:ii}  =  \epsilon(\theta_i, \theta_j) + (C+1)/(2 \deg{G'}) \cdot \sin(\alpha(\theta_i,\theta_j)) , 
\end{align}
is therefore attained at the minimal value of $\epsilon(\theta_i,\theta_j)$ subject to the constraint $|\theta_i| +  |\theta_j| = \alpha \leq \pi/2$. 
Covnersely, the value is attained at the minimal value of $\alpha(\theta_i, \theta_j)$ subject to the constraint \eqref{eq:condition i}. 
This reduces the problem to a one dimensional problem, which we exploit explicitly in the following lemma. 
The intuition behind this lemma is shown in Fig.~\ref{fig:minimize arctan}.
\begin{lemma}
\label{lem:minimal value}
For any fixed value $\pi/2 \geq |\theta_i| + |\theta_j| = \alpha \geq 0 $, the minimal value $\overline \epsilon(\alpha)$ of $\epsilon(\theta_i,\theta_j)$ is achieved at $|\theta_i| = |\theta_j| = \alpha/2$. 
Moreover, for every $\theta_i, \theta_j$ such that $\epsilon(\theta_i, \theta_j) \geq \overline \epsilon(\alpha)$ it holds that $|\theta_i| + |\theta_j| \geq \alpha$.
\end{lemma}
\begin{proof}
	Let $|\theta_i| = (\alpha - \delta)/2$, $|\theta_j| = (\alpha +\delta)/2$ for $0 \leq \delta \leq \pi/2$. Then 
	\begin{align}
		\epsilon(\theta_i, \theta_j) & = -  C \sin\left( \frac{\alpha - \delta}2 \right )\sin\left( \frac{\alpha + \delta}{2} \right) + \cos\left (\frac {\alpha + \delta} 2\right )\cos\left ( \frac{\alpha - \delta}2 \right ) \\
		& = \frac C2 \left( \cos \alpha- \cos \delta \right) + \frac 12 \left( \cos \alpha + \cos \delta \right)\\
		& = \frac12 \left((C + 1 ) \cos \alpha - (C-1) \cos \delta\right) , 
	\end{align}
	which is minimal at $\delta = 0 $. 

	The second part of the lemma can be be seen by contraposition: 
	Assume $|\theta_i| + |\theta_j| < \alpha(\overline\epsilon)$.
	Then
	\begin{align}
		\epsilon(\theta_i, \theta_j) & = \frac12 \left((C + 1 ) \cos( |\theta_i| + |\theta_j| )- (C-1) \cos (|\theta_i| - |\theta_j|)\right) \\
		& \leq \frac12 (C + 1 ) \cos( |\theta_i| + |\theta_j| ) < \frac12 (C + 1 ) \cos \alpha \equiv \overline \epsilon(\alpha) , 
	\end{align}
	where we have used the assumption and the monotonicity of the cosine in the interval $[0,\pi/2]$ in the last inequality and its non-negativity in the interval $[-\pi/2, \pi/2]$ in the second to last inequality.
\end{proof}

Now, for every choice of $\theta_i, \theta_j$, the minimal value of $ \overline{\epsilon}(\alpha)$ of $\epsilon(\theta_i, \theta_j)$ corresponding to $\alpha = |\theta_i| + |\theta_j|$ is attained at $|\theta_i| = |\theta_j| = \alpha/2$. 
Correspondingly, we can re-express $\alpha$ in terms of $\overline \epsilon$ as
\begin{align}
\label{def:alpha(epsilon)}
	\overline \epsilon \equiv \overline \epsilon(\alpha)  = - C \sin^2(\alpha/2) + \cos^2(\alpha/2) \quad \Leftrightarrow \quad&  \alpha(\overline \epsilon) = 2 \arctan \sqrt{\frac{1 - \overline\epsilon}{ C + \overline \epsilon}} . 
\end{align}
The second part of Lemma~\ref{lem:minimal value} states that for all $\theta_i, \theta_j$ such that $\epsilon(\theta_i, \theta_j) \geq \overline{\epsilon}(\alpha) \geq 0$ we have $|\theta_i| +  |\theta_j| \geq \alpha(\overline \epsilon)$ and consequently $\sin(|\theta_i| +  |\theta_j|) \geq \sin(\alpha(\overline \epsilon))$. 
Given $\theta_i, \theta_j$ and defining $\overline \epsilon \coloneqq \epsilon(\theta_i, \theta_j)$ this implies the lower bound
\begin{align}
	\eqref{eq:i} + \eqref{eq:ii} \ge \overline \epsilon + (C+1) /(2 \deg{G'}) \cdot \sin(\alpha(\overline \epsilon)) , 
\end{align}
where we have used the equivalence \eqref{def:alpha(epsilon)}. 

It remains to lower-bound $\sin(\alpha(\overline \epsilon))$. Define $x = \sqrt{(1 - \overline\epsilon)/(C + \overline \epsilon)}$. We can then rewrite
\begin{align}
	\sin( \alpha(\overline \epsilon)) = \sin \left( 2 \arctan x \right) = 2 \sin(\arctan x ) \cos(\arctan x)
	& = 2  \frac{ x} { 1 + x^2} \geq x , 
\end{align}
for $x \leq 1$, where we have used that $\sin(\arctan(x)) = x \cos( \arctan x) = x /\sqrt{1 + x^2}$.
We can also bound 
\begin{align}
	x = \sqrt{\frac{1 - \overline\epsilon}{ C + \overline \epsilon}} \ge \sqrt {\frac 1{C }} \sqrt{\frac{1 - \overline\epsilon}{ 1 + \overline \epsilon/C}} \ge \sqrt{\frac 1{C }} \sqrt{\frac{1 - \overline\epsilon}{ 1 + \overline \epsilon}} \geq \frac{1 - \overline \epsilon}{\sqrt C }, 
\end{align}
where the last inequality can be seen by squaring both sides and using $0 \leq \overline \epsilon < 1 $. 
Combining everything we obtain
\begin{align}
	\eqref{eq:i} + \eqref{eq:ii} \geq  \overline \epsilon + \frac{\sqrt{C}} {2 \deg{G'}} \cdot \left( 1 - \overline \epsilon \right) = \overline \epsilon \left( 1 -   \frac{\sqrt C }{2\deg{G'}} \right) + \frac{\sqrt{C} }{2 \deg G'}    = 1. 
\end{align}
due to our choice of $C = (2 \deg G')^2 $. 

To conclude the proof, we discuss the effect of applying $X$-flips to each of the sites. 
Applying $X_i X_j$ or $X_{\xi_{i,j}}$ merely alters the signs of the terms in \eqref{eq:iii}. 
But since we did not constrain the sign of $\theta_{\xi_{i,j}}$ in the proof, everything remains unchanged. 
Suppose an $X$-flip is applied to either qubit $i$ or $j$, or either both qubit $i$ and $\xi_{i,j}$ or either both qubit $j$ and $\xi_{i,j}$ . 
Assuming wlog.\ that qubit $i$ is $X$-flipped, we achieve the same lower bounds as before by identifying $\theta_i \mapsto  - \theta_i$.

\end{proof}


\end{appendix}

\ifjournal
\putbib
\end{bibunit}
\fi

\end{document}